\documentclass{article}
\usepackage[english]{babel}
\usepackage[utf8]{inputenc}
\usepackage{johd}
\usepackage{color}
\usepackage{bm}
\usepackage{multicol}
\usepackage{amsmath,amssymb,amsthm}
\usepackage{subfigure}
\usepackage{multirow}
\usepackage{float}
\usepackage[a4paper,margin=2.1cm]{geometry}
\newtheorem{theor}{Theorem}
\newtheorem{lemma}{Lemma}
\newtheorem{corollary}{Corollary}

\title{Bivariate generalized autoregressive models for forecasting bivariate non-Gaussian times series}

\author{Tatiane Fontana Ribeiro$^{1,2}$$^{*}$, Airlane P. Alencar$^{1}$, 
	 Fábio M. Bayer$^{2,3,4}$ \\
        \small $^{1}$Institute of Mathematics and Statistics, University of São Paulo, São Paulo, SP, Brazil \\
        \small $^{2}$Graduate Program in Mathematics, Federal University of Santa Maria, Santa Maria, RS, Brazil \\
        \small $^{3}$Department of Mathematics and Natural Sciences, Blekinge Institute of Technology, Karlskrona, Blekinge, Sweden \\
        \small $^{4}$ Department of Statistics and LACESM, Federal University of Santa Maria, Santa Maria, RS, Brazil \\
        \small $^{*}$Corresponding author: Tatiane Fontana Ribeiro; \tt{tatianefr@alumni.usp.br}
}

\date{} 

\begin{document}

\maketitle

\begin{abstract} 
\noindent 
This paper introduces a novel approach, the bivariate generalized autoregressive (BGAR) model, for modeling and forecasting bivariate time series data. The BGAR model generalizes the bivariate vector autoregressive (VAR) models by allowing data that does not necessarily follow a normal distribution.
We consider a random vector of two time series and assume each belongs to the canonical exponential family, similarly to the univariate generalized autoregressive moving average (GARMA) model. We include autoregressive terms 
of one series into the dynamical structure of the other and vice versa. 
The model parameters are estimated using the conditional maximum likelihood (CML) method. We provide general closed-form expressions for the conditional score vector and conditional Fisher information matrix, encompassing all canonical exponential family distributions. We develop asymptotic confidence intervals and hypothesis tests. We discuss techniques for model selection, residual diagnostic analysis, and forecasting. We carry out Monte Carlo simulation studies to evaluate the performance of the finite sample CML inferences, including point and interval estimation. An application to real data analyzes the number of leptospirosis cases on hospitalizations due to leptospirosis in São Paulo state, Brazil. Competing models such as GARMA, autoregressive integrated moving average (ARIMA), and VAR models are considered for comparison purposes. The new model outperforms the competing models by providing more accurate out-of-sample forecasting and allowing quantification of the lagged effect of the case count series on hospitalizations due to leptospirosis.
  \end{abstract}

\noindent\\

\noindent \textbf{Mathematics Subject Classification (MSC):} 62M09, 62M10, 62P10, 62F12, 62H12.
 

\section{Introduction}

Time series analysis plays an essential role in understanding the behavior and dynamics of evolving data over time. Decision makers often need economic, educational, epidemiological, public health variables predictions, and so on~\citep{lutkepohl2005new,shumway2017}. 
Traditional models such as autoregressive moving average (ARMA), autoregressive integrated moving average (ARIMA)~\citep{box2015time}, and its extensions~\citep{granger1980introduction,brockwell2002introduction} are commonly used for univariate time series analysis. These models are usually based on the assumption of linear combinations of variables distributed according to the Gaussian distribution~\citep{chuang2007order}. 

Since the Gaussianity assumption can be unsuitable for many real-world time series data applications, several other models that relax this assumption have been proposed.
\cite{benjamin2003} introduced the generalized autoregressive moving average (GARMA) models that extend the classical ARMA framework to accommodate non-Gaussian data and model a function of the expected response. 
They have been applied in various fields, mainly epidemiology, and public health; see~\cite{dugas2013influenza,talamantes2007statistical,albarracin2018cusum, esparza2018effect}.
Recently, it was proposed various like-GARMA models assuming other distributions for the dependent variable to analyze count data~\citep{melo2020conway,sales2022berg}, double-bounded data
\citetext{see \citealp{rocha2009beta}; 
\citealp{bayer2017kumaraswamy};
 \citealp{bayer2023inflated}; 
\citealp{pumi2023unit}; \\
\citealp{ribeiro2024forecasting}},
binary data \citep{angelo2019bernoulli},
quantized amplitude data and bounded count data \citep{palm2021signal},
and positive time series
\citetext{\citealp{bayer2020novel}; \citealp{almeida2021arma}}.

In addition to univariate time series analysis, understanding the relationships between two or more variables over time is essential for a more comprehensive and accurate analysis. 
In many situations, the value of a variable is related to its predecessors in time and also depends on past
values of other variables. Multiple time series analysis allows investigating several questions regarding the structure of the relationships between the variables involved and obtaining more accurate forecasts~\citep{lutkepohl2005new}. Typically, vector autoregressive (VAR) models are used for the analysis of multivariate time series. 
In economics, these models gained popularity due to Sims' work.~\cite{sims1980macroeconomics, sims1981autoregressive} pioneered innovation accounting and impulse response analysis within VAR models, introducing an alternative to traditional macroeconomic analyses. The framework of VAR models
is one of the most successful and ﬂexible to model multivariate time series~\citep{zivot2006vector}, and it has been applied in many fields.
\cite{shahin2014vector} used the VAR model to forecast monthly temperature, humidity, and cloud coverage of Rajshahi district in Bangladesh. \cite{somosi2024examination}
studied the relationship among economic growth, rising energy consumption, and CO$_2$ emissions, examining whether it is mutual, bidirectional, or unidirectional via VAR models.
This approach also is considered by~\cite{murakami2024difference}
for distinguishing between therapy and antibiotic spectrum coverage days in an inpatient antimicrobial stewardship program.
~\cite{lindberg2024interest} fitted the VAR model to investigate how interest rates affect lifestyle-driven transitions of older adults into independent living properties.

One limitation of the VAR model is that, similarly to ARMA models, the Gaussianity assumption is also required for the response variables. 
Considering other distributions for the data yields improvements in fit and forecasting over a Gaussian specification since the
distribution of many variables can present characteristics including leptokurtosis, fat tails, and skewness~\citep{liu2008statistical,gong2019non,kiss2023modelling}.
Moreover, models based on Gaussianity assumption often poorly describe count and other discrete-valued time series~\citep{davis2021count}.
Some recent works proposed variations of the VAR model assuming other continuous distributions for the error like skew-$t$ normal and skewed generalized $t$-distributions; see for instance~\cite{nduka2024estimation,anttonen2024statistically}.

This paper aims to introduce the bivariate generalized autoregressive (BGAR) model that is more flexible and simpler than the formulation of the bivariate VAR model. 
Considering a structure similar to the univariate GARMA time series model, we additively include autoregressive terms associated with the mean of one series into the other and vice versa, into each linear predictor. 
This allows the model to capture the shock effect of one time series on the other.
Here, we exclude the moving average terms only to improve the interpretation of the parameters.
The BGAR model offers greater simplicity by allowing the autoregressive orders to be selected individually for each series and cross-effects, represented as BGAR($p_{11}, p_{12}, p_{22}, p_{21})$.
Furthermore, the proposed bivariate model can be used to analyze various time-dependent response variables. It can accommodate continuous or discrete data distributed according to the canonical exponential family and different conditional distributions can be combined, accounting for the particular characteristics of each series.

We present the main properties of the BGAR model, the conditional maximum likelihood estimator (CMLE), and provide some tools for model selection and residual analysis. We obtain general closed-form expressions for the conditional score vector and conditional Fisher information matrix, which include all exponential family distributions. 
The CMLE performance is assessed via Monte Carlo simulation studies, in which we combine different family exponential distributions and several autoregressive orders of the models. An application to two count time series is presented and discussed. We analyze the relationship between the number of leptospirosis cases and hospitalizations due to this disease and assess the accuracy of forecasts 12 months ahead.

The paper unfolds as follows. Section~\ref{sec_model_definition} defines the BGAR model. Section~\ref{sec_par_estimation} proposes the conditional likelihood method for parameters estimation and presents the conditional score vector. Large sample inference is discussed and the conditional Fisher information matrix is computed in Section~\ref{sec_large_sample_inf}. In Section~\ref{sec_properties}, we present stationarity conditions for the marginal mean of the BGAR model.
Section~\ref{sec_diag_analysis} contains techniques for model selection, diagnostic analysis, and forecasting.
Section~\ref{sec_simu} presents some Monte Carlo simulation results. 
Section~\ref{sec_app} analyzes the shock effect of the cases on hospitalizations due to leptospirosis in São Paulo state, Brazil. Concluding remarks are presented in Section~\ref{sec_conclusion}.
Finally, Components of some exponential family distributions in the canonical form are provided in Appendix~\ref{sec_append_exp_fam_elements}.
\section{Model definition}\label{sec_model_definition}

Let $\bm{y}_t=(y_{1t},y_{2t})^\top$ be a random vector of two variables with $t\in \mathbb{Z}$ and mean
$\mathrm{I\!E}(\bm{y}_t)=\bm{\mu}_t$, where $\bm{\mu}_t=(\mu_{1t},\mu_{2t})^\top$. Suppose that the random component $y_{kt}$ ($k=1,2$) given the previous information set $\mathcal{F}_{t-1}=\{\bm{y}_{t-1},\ldots,\bm{y}_{1}\}$ has distribution that belongs to the exponential family in the canonical form. That is, we assume that $y_{kt}\vert\mathcal{F}_{t-1}\sim f_k(\mu_{kt},\varphi_k)$ has univariate conditional density $f_k(\cdot)$ expressed as
\begin{align}\label{EQ_pdf_expfam}
	  f_k(y_{kt};\mu_{kt},\varphi\vert \mathcal{F}_{t-1}) = \exp\left\{\frac{y_{kt}\vartheta_{kt} - b_k(\vartheta_{kt})}{\varphi_k} + c_k(y_{kt},\varphi_k)\right\},
\end{align}
where $\vartheta_{kt}$ and $\varphi_k$ are the canonical and dispersion 
parameters, respectively, as in~
\citetext{\citealp{nelder_Wedderburn_1972_glm}; \citealp{benjamin2003}}.

 The conditional mean and variance of the response variable $y_{kt}$ are
given by
$\mu_{kt}=\mathrm{I\!E}(y_{kt}\vert\mathcal{F}_{t-1})=b_k^\prime(\vartheta_{kt})$ and 
$\mathrm{var}(y_{kt}\vert\mathcal{F}_{t-1})=\varphi_kb_k^{\prime\prime}(\vartheta_{kt})=\varphi_kV(\mu_{kt})$, respectively.
The functions $b_k(\cdot)$ and $c_k(\cdot)$ define a particular distribution from the exponential family. 
Here, we consider the Binomial, Negative Binomial, Poisson, Gamma, Normal, and Inverse Normal distributions~\citep{mccullagh1989generalized}. 
All the components of each one of these distributions, represented in the canonical exponential family form, are detailed in Appendix~\ref{sec_append_exp_fam_elements}.

Similarly to the univariate GARMA models~\citep{benjamin2003}, we propose that the mean parameter \( \mu_{kt} \) is related to a linear predictor \( \eta_{kt} \) through a twice-differentiable, one-to-one monotonic link function \( g_k(\cdot) \). Additionally, we incorporate lagged terms \( y_{1t-l} \) and \( y_{2t-l} \) (where \( l \) denotes the lag order) to model \( \mu_{kt} \), capturing both the serial correlation in \( y_{kt} \) and the dependence between \( y_{1t} \) and the lagged values of \( y_{2t} \), as well as between \( y_{2t} \) and the lagged values of \( y_{1t} \).
The dynamical structure of the bivariate model is given by
\begin{align}\label{EQ_model_mu}
&	g_1(\mu_{1t})=\eta_{1t}=
	\bm{x}_{1t}^\top\bm{\beta}_1+
	\sum_{l=1}^{p_{11}}\phi_{11,l}
	[g_1(y_{1t-l})-\bm{x}_{1t-l}^\top\bm{\beta}_{1}]+
		\sum_{l=1}^{p_{12}}
		\phi_{12,l}
	[g_2(y_{2t-l})-\bm{x}_{2t-l}^\top\bm{\beta}_2]
	\\ \nonumber
	&
		g_2(\mu_{2t})=\eta_{2t}=
	\bm{x}_{2t}^\top\bm{\beta}_2+
	\sum_{l=1}^{p_{22}}\phi_{22,l}
	[g_2(y_{2t-l})-\bm{x}_{2t-l}^\top\bm{\beta}_2]+
	\sum_{l=1}^{p_{21}}
	\phi_{21,l}
	[g_1(y_{1t-l})-\bm{x}_{1t-l}^\top\bm{\beta}_1],
\end{align}
where 
$\bm{x}_{kt}=(1,x_{kt,2},\ldots, x_{kt,l_k})^\top\in \mathrm{I\!R}^{l_k}$ is the non-random vector of covariates associated to the time series $y_{kt}$ and 
to the vector of unknown linear parameters
$\bm{\beta}_k=(\beta_{k1},\ldots,\beta_{kl_k})^\top\in \mathrm{I\!R}^{l_k}$ for $k=1,2$.
The vectors 
$\bm{\phi}_{11}=(\phi_{11,1}, \ldots,\phi_{11,p_{11}})^\top$,
$\bm{\phi}_{12}=(\phi_{12,1}, \ldots,\phi_{12,p_{12}})^\top$,
$\bm{\phi}_{22}=(\phi_{22,1}, \ldots,\phi_{22,p_{22}})^\top$,
$\bm{\phi}_{21}=(\phi_{21,1}, \ldots,\phi_{21,p_{21}})^\top$
are the autoregressive parameters,
where
$\phi_{ij,l}\in \mathrm{I\!R}$, $i,j=1,2$, is the autoregressive coefficient of order $p_{ij}$ in lag $l$ that measures the shock effect of the lagged variable $y_{jt-l}$ in the variable $y_{it-l}$.
Therefore, the BGAR$(p_{11},p_{12},p_{22},p_{21})$ model is defined by~\eqref{EQ_pdf_expfam} and~\eqref{EQ_model_mu}.

Like VAR models, the shock effects are useful in the BGAR model as they allow the dynamic interaction between the two series to be captured. Each coefficient 
$\phi_{ij,l}$ quantifies how past values of one variable affect the current value of the other, making possible to model not only the temporal dependencies within each series but also the cross-dependencies between them.

\section{Parameter estimation}\label{sec_par_estimation}

Parameter estimation is carried out by conditional maximum likelihood~\citep{andersen1970asymptotic}.
Let $y_{k1}, \ldots,y_{kn}$, $k=1,2$,
be a sample from a BGAR$(p_{11},p_{12},p_{22},p_{21})$ model under specification~\eqref{EQ_pdf_expfam} and~\eqref{EQ_model_mu}.
Let 
$\bm{\gamma}=
(\bm{\beta}_1^\top,
\bm{\beta}_2^\top,
\bm{\phi}_{11}^\top,
\bm{\phi}_{12}^\top$,
$\bm{\phi}_{22}^\top,
\bm{\phi}_{21}^\top,
$
$
\varphi_1,
\varphi_2
)^\top$ 
be the
parameter vector that index 
this
model.
The dimension of $\bm{\gamma}$
depends on the marginal distributions assumed for each random component in the model specification. When both components follow distributions with dispersion parameters (e.g., Normal, Gamma, or Inverse Normal), 
the total number of parameters is $v = l_1 + l_2 + p_{11} + p_{12} + p_{22} + p_{21} + 2$. On the other hand, if both components follow distributions with fixed or absent dispersion (e.g., Negative Binomial, Binomial or Poisson), 
the reduced vector $\lambda = (\beta_1^\top, \beta_2^\top, \boldsymbol{\varphi}_{11}^\top, \boldsymbol{\varphi}_{12}^\top, \boldsymbol{\varphi}_{22}^\top, \boldsymbol{\varphi}_{21}^\top)^\top$ 
is used, with dimension $v^\star = l_1 + l_2 + p_{11} + p_{12} + p_{22} + p_{21}$.
In mixed cases, where only one component has a dispersion parameter (e.g., $y_{1t} \sim$ Normal and $y_{2t} \sim$ Poisson), the parameter vector $\gamma$ has dimension $v^\dagger = l_1 + l_2 + p_{11} + p_{12} + p_{22} + p_{21} + 1$. 
Observe that $y_{1t}|\mathcal{F}_{t-1}$ and $y_{2t}|\mathcal{F}_{t-1}$
are independent, as they are conditioned only on the previous information set.
The likelihood function for the parameter vector $\bm{\gamma}$ conditional on the first  $m=\mathrm{max}(p_{11},p_{12},p_{22},p_{21})$ observations can be expressed as
\begin{align}\label{EQ_likelihood}
	L(\bm{\gamma})=
	 \prod_{k=1}^{2}
	\prod_{t=m+1}^{n}
	f_k(y_{kt}\vert \mathcal{F}_{t-1}).
\end{align}
Taking the logarithm on both sides of~\eqref{EQ_likelihood}, we write the conditional log-likelihood function as
\begin{align}\label{EQ_loglik}
	\ell(\bm{\gamma})=
	\sum_{k=1}^{2}
	\sum_{t=m+1}^{n}
	\mathrm{log}(f_k(y_{kt}\vert \mathcal{F}_{t-1}))=
	\sum_{k=1}^{2}
	\sum_{t=m+1}^{n}
	\ell_{kt}(\vartheta_{kt},\varphi_k),
\end{align}
where
\begin{align*}
	\ell_{kt}(\vartheta_{kt},\varphi_k)=
	\frac{y_{kt}\vartheta_{kt} - b_k(\vartheta_{kt})}{\varphi_k} + c_k(y_{kt},\varphi_k).
\end{align*}
The CMLE of $\bm{\gamma}$ is obtained
by maximizing~\eqref{EQ_loglik} assuming that the number of parameters is $h<(n-m)$.

\subsection{Conditional score vector}\label{sec_score_vector}


The score vector $\bm{U}(\bm{\gamma})$ consists of the partial derivatives of the conditional log-likelihood function $\ell:= \ell(\bm{\gamma})$ given in~\eqref{EQ_loglik} with respect to each component of $\bm{\gamma}$. 
Let $\bm{\lambda} = (
\bm{\beta}_1^\top,
\bm{\beta}_2^\top,
\bm{\phi}_{11}^\top,
\bm{\phi}_{12}^\top,
\bm{\phi}_{22}^\top,
\bm{\phi}_{21}^\top
)^\top$.
In what follows, we shall obtain the derivatives of $\ell$ 
 with respect to the $i$th element of $\bm{\lambda}$, $\lambda_i$, where $i = 1,2,\ldots,v^\star$.
By applying the chain rule~\citep{nelder_Wedderburn_1972_glm}, we obtain
\begin{align}\label{EQ_score1}
	U_{\lambda _i}(\bm{\gamma})=
	\frac{\partial \ell}{\partial \lambda_i}=
		\sum_{k=1}^2
	\sum_{t=m+1}^n 
	\frac
	{\partial \ell_{kt}(\vartheta_{kt},\varphi_k)}
	{\partial \vartheta_{kt}}
	\frac{\mathrm{d} \vartheta_{kt}}{\mathrm{d}\mu_{kt} }
	\frac{\mathrm{d} \mu_{kt}}{\mathrm{d}\eta_{kt} }
	\frac{\partial \eta_{kt}}{\partial \lambda_i}.
\end{align}
The derivative of the first factor in~\eqref{EQ_score1} reduces to
\begin{align}\label{EQ_ell_vartheta}
		\frac
	{\partial \ell_{kt}(\vartheta_{kt},\varphi_k)}
	{\partial \vartheta_{kt}}=
		\frac{y_{kt} - b_k^\prime(\vartheta_{kt})}
		{\varphi_k}=	\frac{y_{kt} - \mu_{kt}}
		{\varphi_k},
\end{align}
since $\mu_{kt}=b_k^\prime(\vartheta_{kt})$.
Furthermore, the derivative of second and third factors in~\eqref{EQ_score1} are given by
\begin{align*}
		\frac{\mathrm{d} \vartheta_{kt}}{\mathrm{d}\mu_{kt} }=
		\frac{1}{b_k^{\prime\prime}(\vartheta_{kt})}=
		V_{kt}^{-1},
		\qquad
		\text{and}
		\qquad
		\frac{\mathrm{d} \mu_{kt}}{\mathrm{d}\eta_{kt} }=
		\frac{1}{g^\prime(\mu_{kt})},	
\end{align*}
where $V_{kt}:=V(\mu_{kt})=
\text{d}\mu_{kt}/\text{d}\vartheta_{kt}$. Note that this definition implies $b_k^{\prime\prime}(\vartheta_{kt})=V_{kt}$.
Besides, following the notation of the generalized linear models, we also define the weights $\omega_{kt}=(\mathrm{d} \mu_{kt}/\mathrm{d}\eta_{kt})^2/V_{kt}$.

By replacing the above quantities in~\eqref{EQ_score1}, we have that
\begin{align}\label{EQ_score2}
	\nonumber
	U_{\lambda _i}(\bm{\gamma})=
	\frac{\partial \ell}{\partial \lambda_i}
	&= 
	\sum_{k=1}^2
	\frac{1}{\varphi_k}
	\sum_{t=m+1}^n 
	\sqrt{\frac{\omega_{kt}}{V_{kt}}}\,
(y_{kt}-\mu_{kt})
	\frac{\partial \eta_{kt}}{\partial \lambda_i}
	\\ 
	&=
		\frac{1}{\varphi_1}
	\sum_{t=m+1}^n 
	\sqrt{\frac{\omega_{1t}}{V_{1t}}}\,
	(y_{1t}-\mu_{1t})
	\frac{\partial \eta_{1t}}{\partial \lambda_i}
	+
			\frac{1}{\varphi_2}
	\sum_{t=m+1}^n 
	\sqrt{\frac{\omega_{2t}}{V_{2t}}}\,
	(y_{2t}-\mu_{2t})
	\frac{\partial \eta_{2t}}{\partial \lambda_i}.
\end{align}

The derivatives 
$\partial \eta_{1t}/\partial \lambda_i$ and
$\partial \eta_{2t}/\partial \lambda_i$
in~\eqref{EQ_score2}
are given by 
\begin{align*}
	\frac
	{\partial \eta_{1t}}
	{\partial \beta_{1s}}
	=
	x_{1t,s}-
	\sum_{i=1}^{p_{11}}\phi_{11,l}
{x}_{1t-l,s}
	\qquad
\text{for}
\qquad
s=1,2,\ldots,l_1,
\end{align*}
\begin{align*}
		\frac
	{\partial \eta_{2t}}
	{\partial \beta_{1s}}
	=
	-
	\sum_{l=1}^{p_{21}}
	\phi_{21,l}
	{x}_{1t-l,s}
	\qquad
	\text{for}
	\qquad
	s=1,2,\ldots,l_1,
\end{align*}
\begin{align*}
	\frac
	{\partial \eta_{1t}}
	{\partial \beta_{2w}}
	=
-\sum_{l=1}^{p_{12}}
\phi_{12,l}
{x}_{2t-l,w}
	\qquad
\text{for}
\qquad
w=1,2,\ldots,l_2,
\end{align*}
\begin{align*}
		\frac
	{\partial \eta_{2t}}
	{\partial \beta_{2w}}
	=
	x_{2t,w}-
	\sum_{l=1}^{p_{22}}\phi_{22,l}
	{x}_{2t-l,w}
	\qquad
	\text{for}
	\qquad
	w=1,2,\ldots,l_2,
\end{align*}
\begin{align*}
	\frac
	{\partial \eta_{1t}}
	{\partial \phi_{11,l}}
	=
g(y_{1t-l})-\bm{x}_{1t-l}^\top\bm{\beta_{1}}
	\qquad
	\text{for}
	\qquad
	l=1,2,\ldots,p_{11},
\end{align*}
\begin{align*}
	\frac
	{\partial \eta_{1t}}
	{\partial \phi_{12,l}}
	=
	g(y_{2t-l})-\bm{x}_{2t-l}^\top\bm{\beta_{2}}
	\qquad
	\text{for}
	\qquad
	l=1,2,\ldots,p_{12},
\end{align*}
\begin{align*}
	\frac
	{\partial \eta_{2t}}
	{\partial \phi_{22,l}}
	=g(y_{2t-l})-\bm{x}_{2t-l}^\top\bm{\beta_{2}}
	\qquad
	\text{for}
	\qquad
	l=1,2,\ldots,p_{22},
\end{align*}
and
\begin{align*}
	\frac
	{\partial \eta_{2t}}
	{\partial \phi_{21,l}}
	=g(y_{1t-l})-\bm{x}_{1t-l}^\top\bm{\beta_{1}}
	\qquad
	\text{for}
	\qquad
	l=1,2,\ldots,p_{21}.
\end{align*}
Note that the derivatives 
$	\partial \eta_{2t}/\partial \phi_{11,l}$
($l = 1,\ldots,p_{11}$),
$	\partial \eta_{2t}/\partial \phi_{12,l}$
($l = 1,\ldots,p_{12}$),
$	\partial \eta_{1t}/\partial \phi_{21,l}$
($l = 1,\ldots,p_{21}$),
and
$	\partial \eta_{1t}/\partial \phi_{22,l}$
($l = 1,\ldots,p_{22}$),	
are equal to zero.

Finally, the last components of the conditional score vector, 
$	U_{\varphi_k}(\bm{\gamma})$,
follows from direct differentiation of~\eqref{EQ_loglik} and can be expressed as 
\begin{align}\label{EQ_dell_dvarphik}
	U_{\varphi_k}(\bm{\gamma})=
\frac{
\partial \ell
}
{\partial\varphi_k}
=
\sum_{t=m+1}^n
\frac{\partial \ell_{kt} (\vartheta_{kt},\varphi_k)}{\partial\varphi_k}
=
\sum_{t=m+1}^n
a_{kt},
\end{align}
where
\begin{align*}
	a_{kt}=-\varphi_k^{-2}[y_{kt}\vartheta_{kt} - b_k(\vartheta_{kt})] + c_k^\prime(y_{kt},\varphi_k)
\end{align*}
and
$$c_k^\prime(y_{kt},\varphi_k)=\frac{\mathrm{d}\, c_k(y_{kt},\varphi_k)}{\mathrm{d} \varphi_k}.$$ 
The quantities $a_{kt}$ and $c_k^\prime(\cdot,\cdot)$ are computed in Appendix~\ref{sec_append_exp_fam_elements} for the Binomial, Poisson, Negative Binomial, Gamma, Normal, and Inverse Normal distributions.

To simplify the notation, we define the matrices
$
\bm{B}_{11},
\bm{B}_{21},
\bm{B}_{12},
\bm{B}_{22},
\bm{P}_{11},
\bm{P}_{12},
\bm{P}_{22},
\,\,
\text{and}
\,\,
\bm{P}_{21},
$
with dimensions 
$
(n-m)\times l_1,\,
(n-m)\times l_1,\,
(n-m)\times l_2,\,
(n-m)\times l_2,\,
(n-m)\times p_{11},\,
(n-m)\times p_{12},\,
(n-m)\times p_{22},\,
$
and
$
(n-m)\times p_{21},\,
$ respectively. 
The $(i,j)$th element of each one these matrices is given by
\begin{align*}
	&	{B}_{11_{i,j}}=
	\frac
	{\partial \eta_{1_{i+m}}}
	{\partial \beta_{1_j}},\quad 
	{B}_{21_{i,j}}=
		\frac
	{\partial \eta_{2_{i+m}}}
	{\partial \beta_{1_j}},\quad 
			{B}_{12_{i,j}}=
	\frac
	{\partial \eta_{1_{i+m}}}
	{\partial \beta_{2_j}},\quad 
	{B}_{22_{i,j}}=
	\frac
	{\partial \eta_{2_{i+m}}}
	{\partial \beta_{2_j}},
	\quad 
	\\ &
	P_{11_{i,j}}=\frac{\partial \eta_{1_{i+m}}}{\partial \phi_{11,j}}, 
	\quad 
	P_{12_{i,j}}=\frac{\partial \eta_{1_{i+m}}}{\partial \phi_{12,j}}, 
	\quad 
	P_{22_{i,j}}=\frac{\partial \eta_{2_{i+m}}}{\partial \phi_{22,j}}, 
		\quad 
		\text{and}
		\quad
	P_{21_{i,j}}=\frac{\partial \eta_{2_{i+m}}}{\partial \phi_{21,j}}, 
\end{align*}
respectively. 
By using these quantities,
 we can compactly write the elements of the score vector
 $\bm{U}(\bm{\gamma})=
 (
 \bm{U}_{\bm{\beta}_1} (\bm{\gamma})^\top,
 \bm{U}_{\bm{\beta}_2} (\bm{\gamma})^\top,
 \bm{U}_{\bm{\phi}_{11}} (\bm{\gamma})^\top,
 \bm{U}_{\bm{\phi}_{12}} (\bm{\gamma})^\top,
 \bm{U}_{\bm{\phi}_{22}} (\bm{\gamma})^\top,
 \bm{U}_{\bm{\phi}_{21}} (\bm{\gamma})^\top,
 	{U}_{\varphi_1} (\bm{\gamma}),
 		{U}_{\varphi_2} (\bm{\gamma})
 )^\top$,
 in the matrix form, as 
\begin{align*}
	\bm{U}_{\bm{\beta}_1} (\bm{\gamma})=&\,
	\varphi_1^{-1}
	\bm{B}_{11}^\top
	\bm{W}_1^{\frac{1}{2}}
	\bm{V}_1^{-\frac{1}{2}}
	\,(\bm{y}_1-\bm{\mu}_1)
	+
	\varphi_2^{-1}
	\bm{B}_{21}^\top
	\bm{W}_2^{\frac{1}{2}}
	\bm{V}_2^{-\frac{1}{2}}
	\,(\bm{y}_2-\bm{\mu}_2),
	&\\
	\bm{U}_{\bm{\beta}_2} (\bm{\gamma})=&\,
	\varphi_1^{-1}
	\bm{B}_{12}^\top
	\bm{W}_1^{\frac{1}{2}}
	\bm{V}_1^{-\frac{1}{2}}
	\,(\bm{y}_1-\bm{\mu}_1)
	+
	\varphi_2^{-1}
	\bm{B}_{22}^\top
	\bm{W}_2^{\frac{1}{2}}
	\bm{V}_2^{-\frac{1}{2}}
	\,(\bm{y}_2-\bm{\mu}_2),
	&\\
	\bm{U}_{\bm{\phi}_{11}} (\bm{\gamma})=&\,
	\varphi_1^{-1}
	\bm{P}_{11}^\top
	\bm{W}_1^{\frac{1}{2}}
	\bm{V}_1^{-\frac{1}{2}}
	\,(\bm{y}_1-\bm{\mu}_1),
	&\\
	\bm{U}_{\bm{\phi}_{12}} (\bm{\gamma})=&\,
	\varphi_1^{-1}
	\bm{P}_{12}^\top
	\bm{W}_1^{\frac{1}{2}}
	\bm{V}_1^{-\frac{1}{2}}
	\,(\bm{y}_1-\bm{\mu}_1),
	&\\
	\bm{U}_{\bm{\phi}_{22}} (\bm{\gamma})=&\,
	\varphi_2^{-1}
	\bm{P}_{22}^\top
	\bm{W}_2^{\frac{1}{2}}
	\bm{V}_2^{-\frac{1}{2}}
	\,(\bm{y}_2-\bm{\mu}_2),
	&\\
	\bm{U}_{\bm{\phi}_{21}} (\bm{\gamma})=&\,
	\varphi_2^{-1}
	\bm{P}_{21}^\top
	\bm{W}_2^{\frac{1}{2}}
	\bm{V}_2^{-\frac{1}{2}}
	\,(\bm{y}_2-\bm{\mu}_2),
	&\\
	{U}_{\varphi_1} (\bm{\gamma})=&\,
	\bm{a}_1^\top\bm{1},
	\quad
	\text{and}
	\quad
	&\\
	{U}_{\varphi_2} (\bm{\gamma})=&\,
	\bm{a}_2^\top\bm{1},
\end{align*}
where
$\bm{W}_k=\mathrm{diag}\left\{
\omega_{m+1},
\ldots,
\omega_n\right\}$,
$\bm{V}_k=\mathrm{diag}\left\{
V_{m+1},
\ldots,
V_n\right\}$,
$\bm{y}_k=(y_{k(m+1)},\ldots,y_{kn})^\top$
$\bm{\mu}_k=(\mu_{k(m+1)},\ldots,\mu_{kn})^\top$,
$\bm{a}_k=(a_{km+1},\ldots,a_{kn})^\top$,
and $\bm{1}$ is an $(n - m)$-dimensional vector of ones.

The CMLE $\hat{\bm{\gamma}}$ of $\bm{\gamma}$ is obtained as the solution of the system of equations $\bm{U}(\bm{\gamma})=\bm{0}$, where $\bm{0}$ is an $h\times 1$ vector of zeros. Since this system of nonlinear equations does not have a closed-form solution, nonlinear optimization algorithms are required to obtain the conditional maximum likelihood estimates by numerically maximizing the conditional log-likelihood function. We use the quasi-Newton algorithm known as Broyden-Fletcher-Goldfarb-Shanno (BFGS); for more details, the reader is referred to~\cite{press1992}.
This method requires initialization since it is an iterative optimization algorithm. To obtain the starting values of $\bm{\gamma}$, we fit a generalized linear model (GLM) for each time series, where the response variables are 
	$\boldsymbol{y}_k=(y_{k(m+1)}, \ldots, y_{kn})^\top$
and the covariate matrices are expressed as 
\begin{align*}
		\boldsymbol{X}_1=
		\left[
		\begin{array}{@{}*{12}{c}@{}}
			1& x_{1(m+1)1} & \ldots & x_{1(m+1)l_1} & y^\star_{1m}& y^\star_{1(m-1)}& \cdots & y^\star_{1(m-p_{11}+1)}\\
			1& x_{1(m+2)1} & \cdots & x_{1(m+2)l_1} & y^\star_{1(m+1)}& y^\star_{1m}& \cdots & y^\star_{1(m-p_{11}+2)}\\
			\vdots & \vdots & \ddots & \vdots & \vdots & \vdots & \ddots & \vdots  \\
			1& x_{1n1} & \cdots & x_{1nl_1} & y^\star_{1(n-1)}& y^\star_{1(n-2)}& \cdots & y^\star_{1(n-p_{11})}\\
		\end{array}
\right]	
\end{align*}
and
\begin{align*}
	\boldsymbol{X}_2=
	\left[
	\begin{array}{@{}*{12}{c}@{}}
		1& x_{2(m+1)1} & \ldots & x_{2(m+1)l_2} & y^\star_{2m}& y^\star_{2(m-1)}& \cdots & y^\star_{2(m-p_{22}+1)}\\
		1& x_{2(m+2)1} & \cdots & x_{2(m+2)l_2} & y^\star_{2(m+1)}& y^\star_{2m}& \cdots & y^\star_{2(m-p_{22}+2)}\\
		\vdots & \vdots & \ddots & \vdots & \vdots & \vdots & \ddots & \vdots  \\
		1& x_{2n1} & \cdots & x_{2nl_2} & y^\star_{2(n-1)}& y^\star_{2(n-2)}& \cdots & y^\star_{2(n-p_{22})}\\
	\end{array}
	\right],
\end{align*}
where $y^\star_{kt}=g_k(y_{kt})$ for $k=1,2$ as in~\cite{bayer2017kumaraswamy,ribeiro2024forecasting,bayer2025novel}.
The starting values for $\bm{\beta}_{1},\bm{\phi}_{11}$, and $\varphi_1$ are the parameter estimates from a GLM with response $\bm{y}_1$ and covariates matrix $\bm{X}_1$.
Similarly, the initial values for $\bm{\beta}_{2},\bm{\phi}_{22}$, and $\varphi_2$ are the parameter estimates from a GLM with response $\bm{y}_2$ and covariates matrix $\bm{X}_2$. The initial values of $\bm{\phi}_{12}$ and $\bm{\phi}_{21}$ are set to zero.
The function \verb|glm.nb| from the \verb|MASS| package in \verb|R| language programming~\citep{R_ref} is used to fit a Negative Binomial GLM. We use the \verb|glm| function from the \verb|stats| package for other distributions.

Since the $\kappa_k$ must be considered known so that NB distribution belongs to a canonical exponential family, we set $\kappa_k$ as equal to the initial guess of them obtained according to this initial values scheme.
This approach ensures that the NB distribution is in the canonical form, allowing the use of the general quantities derived in this Section and in Section~\ref{sec_fisher}.

\section{Large sample inference}\label{sec_large_sample_inf}

This section aims to present large sample inferences for the proposed model, such as interval estimation and hypothesis testing. Under usual regularity conditions, the CMLE $\bm{\hat{\gamma}}$ of $\bm{\gamma}$ is consistent and converge in distribution to a multivariate normal distribution with mean $\bm{\gamma}$ and an asymptotic covariance matrix given by the inverse of the conditional Fisher information matrix $\bm{K}(\bm{\gamma})$~\citep{andersen1970asymptotic}. That is, 
in large sample sizes, we have that 
\begin{align}\label{EQ_normass_CLME}
	\bm{\hat{\gamma}}\sim  \mathcal{N}_v
	\left(
	\bm{\gamma},\bm{K}^{-1}(\bm{\gamma})
	\right),
\end{align}
approximately, 
and
$\mathcal{N}_v$ denotes the multivariate normal distribution of a $v$-dimensional vector.

\subsection{Conditional Fisher information matrix}\label{sec_fisher}

In what follows, we shall compute the expected values of all second-order derivatives of the conditional log-likelihood function to obtain the matrix $\bm{K}(\bm{\gamma})$.
 
The second derivatives of the conditional log-likelihood with respect to $\lambda_j$ are given by
\begin{align}\label{EQ_secDeriv_loglik}
	\nonumber
	\frac{\partial ^2 \ell}{\partial \lambda_i \partial \lambda_j}
	&=
	\sum_{k=1}^2
	\sum_{t=m+1}^n 
	\frac{\partial^2 \ell_{kt}(\vartheta_{kt},\varphi_{kt})}{\partial \lambda_i \partial \lambda_j}\\
	\nonumber
	&=
	\sum_{k=1}^2
	\sum_{t=m+1}^n 
	\frac{\partial }{\partial \lambda_i}
	\left(
	\frac
	{\partial \ell_{kt}(\vartheta_{kt},\varphi_k)}
	{\partial \vartheta_{kt}}
	\frac{\mathrm{d} \vartheta_{kt}}{\mathrm{d}\mu_{kt} }
	\frac{\mathrm{d} \mu_{kt}}{\mathrm{d}\eta_{kt} }
	\frac{\partial \eta_{kt}}{\partial \lambda_j}
	\right)
	\\ 
	&=
		\sum_{k=1}^2
	\sum_{t=m+1}^n 
	\frac{\partial }{\partial \mu_{kt}}
	\left(
	\frac
	{\partial \ell_{kt}(\vartheta_{kt},\varphi_k)}
	{\partial \vartheta_{kt}}
	\frac{\mathrm{d} \vartheta_{kt}}{\mathrm{d}\mu_{kt} }
	\frac{\mathrm{d} \mu_{kt}}{\mathrm{d}\eta_{kt} }
	\frac{\partial \eta_{kt}}{\partial \lambda_j}
	\right)
		\frac{\mathrm{d} \mu_{kt}}{\mathrm{d}\eta_{kt} }
		\frac{\partial \eta_{kt}}{\partial \lambda_i},
\end{align}
for $i=1,\ldots,v^\star$ and $j=1,\ldots,v^\star$.
By applying the product rule in~\eqref{EQ_secDeriv_loglik}, it follows that
\begin{align}\label{EQ_secDeriv_loglik1}
	\nonumber
	\frac{\partial ^2 \ell}{\partial \lambda_i \partial \lambda_j}
	=
	\sum_{k=1}^2
	\sum_{t=m+1}^n 
	&
	\left[
	\frac
	{\partial ^2\ell_{kt}(\vartheta_{kt},\varphi_k)}
	{\partial \vartheta_{kt}^2}
	\left(
	\frac{\mathrm{d} \vartheta_{kt}}{\mathrm{d}\mu_{kt} }
	\right)^2
	\frac{\mathrm{d} \mu_{kt}}{\mathrm{d}\eta_{kt} }
	\frac{\partial \eta_{kt}}{\partial \lambda_j}
	+
		\frac
	{\partial \ell_{kt}(\vartheta_{kt},\varphi_k)}
	{\partial \vartheta_{kt}}
	\frac{\mathrm{d}^2 \vartheta_{kt}}{\mathrm{d}\mu_{kt}^2 }
	\frac{\mathrm{d} \mu_{kt}}{\mathrm{d}\eta_{kt} }
	\frac{\partial \eta_{kt}}{\partial \lambda_j}
	+
	\right.
	\\
	&
	\left.
		\frac
	{\partial \ell_{kt}(\vartheta_{kt},\varphi_k)}
	{\partial \vartheta_{kt}}
	\frac{\mathrm{d} \vartheta_{kt}}{\mathrm{d}\mu_{kt} }
	\frac{\partial }{\partial \mu_{kt}}
	\left(
	\frac{\mathrm{d} \mu_{kt}}{\mathrm{d}\eta_{kt} }
	\frac{\partial \eta_{kt}}{\partial \lambda_j}
	\right)
		\right]
	\frac{\mathrm{d} \mu_{kt}}{\mathrm{d}\eta_{kt} }
	\frac{\partial \eta_{kt}}{\partial \lambda_i}.
\end{align}

From the usual regularity conditions, we have that
$\mathrm{I\!E}(\partial \ell_{kt}(\vartheta_{kt},\varphi_k)/
\partial \vartheta_{kt}\vert \mathcal{F}_{t-1})=0$. 
Computing the expected values of~\eqref{EQ_secDeriv_loglik1}, we obtain
\begin{align}
	\nonumber
	\mathrm{I\!E}
	\left(
\frac{\partial ^2\ell_{kt}(\vartheta_{kt},\varphi_{kt})}{\partial \lambda_i \partial \lambda_j}
	\bigg\vert
	\mathcal{F}_{t-1}
	\right)
	=\,
	&
		\mathrm{I\!E}
		\left(
	\frac
	{\partial ^2\ell_{kt}(\vartheta_{kt},\varphi_k)}
	{\partial \vartheta_{kt}^2}
		\bigg\vert
	\mathcal{F}_{t-1}
	\right)
	\left(
	\frac{\mathrm{d} \vartheta_{kt}}{\mathrm{d}\mu_{kt} }
	\right)^2
\left(
	\frac{\mathrm{d} \mu_{kt}}{\mathrm{d}\eta_{kt} }
		\right)^2
	\frac{\partial \eta_{kt}}{\partial \lambda_j}
	\frac{\partial \eta_{kt}}{\partial \lambda_i}.
\end{align}

From~\eqref{EQ_ell_vartheta}, we have that
\begin{align*}
	\frac
{\partial ^2\ell_{kt}(\vartheta_{kt},\varphi_k)}
{\partial \vartheta_{kt}^2}
=
-\frac{b_k^{\prime\prime}(\vartheta_{kt})}{\varphi_k}
=
-\frac{V_{kt}}{\varphi_k}.
\end{align*}
On the other hand, observe that
\begin{align*}
		\left(
	\frac{\mathrm{d} \vartheta_{kt}}{\mathrm{d}\mu_{kt} }
	\right)^2
	\left(
	\frac{\mathrm{d} \mu_{kt}}{\mathrm{d}\eta_{kt} }
	\right)^2
	=\frac{w_{kt}}
	{V_{kt}}.
\end{align*}
Hence,
\begin{align*}
	\nonumber
	\mathrm{I\!E}
	\left(
\frac{\partial ^2\ell_{kt}(\vartheta_{kt},\varphi_{kt})}{\partial \lambda_i \partial \lambda_j}
	\bigg\vert
	\mathcal{F}_{t-1}
	\right)
	&=
	-
\frac{w_{kt}}{\varphi_k}
	\frac{\partial \eta_{kt}}{\partial \lambda_j}
	\frac{\partial \eta_{kt}}{\partial \lambda_i}.
\end{align*}
Note that all first derivatives $\partial \eta_{kt}/\partial \lambda_i$ 
(which are equal to
 $\partial \eta_{kt}/\partial \lambda_j$)
have
already been presented in Section~\ref{sec_score_vector}.

The second-order derivatives
 $\partial^2 \ell/\partial \varphi_k^2$ follows from the direct differentiation of~\eqref{EQ_dell_dvarphik} with respect to $\varphi_k$ as
\begin{align}\label{EQ_dell2_dvarphi2}
	\frac{\partial^2 \ell}{\partial \varphi_k^2}
	= 
	\sum_{t=m+1}^n
	\frac{\partial \ell_{kt}^2 (\vartheta_{kt},\varphi_k)}{\partial\varphi_k^2}
	=
	\sum_{t=m+1}^n
	\{
	2\varphi_k^{-3}[y_{kt}\vartheta_{kt} - b_k(\vartheta_{kt})] + c_k^{\prime\prime}(y_{kt},\varphi_k)
	\}.
\end{align}
Taking the expected value of~\eqref{EQ_dell2_dvarphi2}, we obtain
\begin{align}\label{EQ_expect_dell_dvarphi2}
	\mathrm{I\!E}
	\left(
	\frac{\partial \ell_{kt}^2 (\vartheta_{kt},\varphi_k)}{\partial\varphi_k^2}
	\bigg\vert
	\mathcal{F}_{t-1}
	\right)
	= 
	2\varphi_k^{-3}[\mu_{kt}\vartheta_{kt} - b_k(\vartheta_{kt})] +
		\mathrm{I\!E}\left( c_k^{\prime\prime}(y_{kt},\varphi_k)\vert
		\mathcal{F}_{t-1}\right)
	=	 d_{kt},
\end{align}
where 
\begin{align*}
	d_{kt}=	2\varphi_k^{-3}[\mu_{kt}\vartheta_{kt} - b_k(\vartheta_{kt})] +
	\mathrm{I\!E}\left( c_k^{\prime\prime}(y_{kt},\varphi_k)\vert \mathcal{F}_{t-1}\right).
\end{align*}
The quantities $c_k^{\prime\prime}(y_{kt},\varphi_k)$, $d_{kt}$, and $	\mathrm{I\!E}\left( c_k^{\prime\prime}(y_{kt},\varphi_k)\vert \mathcal{F}_{t-1}\right)$ are given in Appendix~\ref{sec_append_exp_fam_elements} for the Binomial, Negative Binomial, Poisson, Gamma, Normal, and Inverse Normal distributions. 

The cross-derivative of the log-likelihood function with respect to $\varphi_1$ and $\varphi_2$ is null. That is,
\begin{align*}
\frac{\partial ^2\ell}{\partial \varphi_1\partial \varphi_2}
= 0.
\end{align*}

Finally,
we shall now obtain the cross derivatives with respect to $\lambda_i$, $i=1,\ldots(h-2)$, and $\varphi_k$. 
Recall that
\begin{align*}
	\frac{\partial \ell}{\partial \lambda_i}
	= 
	\sum_{k=1}^2
	\frac{1}{\varphi_k}
	\sum_{t=m+1}^n 
	\sqrt{\frac{\omega_{kt}}{V_{kt}}}\,
	(y_{kt}-\mu_{kt})
	\frac{\partial \eta_{kt}}{\partial \lambda_i}.
\end{align*}
Therefore,
\begin{align*}
	\frac{\partial^2 \ell}{\partial \lambda_i \partial \varphi_k}
	= 
	-
	\frac{1}{\varphi_k^2}
	\sum_{t=m+1}^n 
	\sqrt{\frac{\omega_{kt}}{V_{kt}}}\,
	(y_{kt}-\mu_{kt})
	\frac{\partial \eta_{kt}}{\partial \lambda_i},
	\qquad k=1,2.
\end{align*}
Since $	\mathrm{I\!E}( y_{kt}\vert \mathcal{F}_{t-1})=\mu_{kt}$, it follows that $\mathrm{I\!E}
	\left(
	\partial^2 \ell/\partial \lambda_i \partial \varphi_k
	\right)
	= 0.$
That is, the parameters
$\bm{\lambda}$ and $\bm{\varphi}$ are orthogonal.
This implies that the conditional Fisher information matrix for $\bm{\gamma}$ is block-diagonal. Therefore,
$\bm{K}(\bm{\gamma})=\mathrm{diag}\{\bm{K}_{{\lambda\lambda}},\bm{K}_{{\varphi\varphi}}\}$.
 The matrix
$\bm{K}_{{\lambda\lambda}}$ is given by
\begin{align*}
	{\bm{K}_{{\lambda\lambda}}=
		\begin{bmatrix}
\bm{K}_{\bm{\beta}_1\bm{\beta}_1} & 
\bm{K}_{\bm{\beta}_1\bm{\beta}_2} &
\bm{K}_{\bm{\beta}_1\bm{\phi}_{11}} &
\bm{K}_{\bm{\beta}_1\bm{\phi}_{12}} &
\bm{K}_{\bm{\beta}_1\bm{\phi}_{22}} &
\bm{K}_{\bm{\beta}_1\bm{\phi}_{21}} \\
		\bm{K}_{\bm{\beta}_2\bm{\beta}_1} & 
		\bm{K}_{\bm{\beta}_2\bm{\beta}_2} &
		\bm{K}_{\bm{\beta}_2\bm{\phi}_{11}} &
		\bm{K}_{\bm{\beta}_2\bm{\phi}_{12}} &
		\bm{K}_{\bm{\beta}_2\bm{\phi}_{22}} &
		\bm{K}_{\bm{\beta}_2\bm{\phi}_{21}} \\
 \bm{K}_{\bm{\phi}_{11}\bm{\beta}_1} & 
\bm{K}_{\bm{\phi}_{11}\bm{\beta}_2} & 
\bm{K}_{\bm{\phi}_{11}\bm{\phi}_{11}} &
\bm{K}_{\bm{\phi}_{11}\bm{\phi}_{12}} &
\bm{0}_{p_{11},p_{22}} &
\bm{0}_{p_{11},p_{21}}
 \\
  \bm{K}_{\bm{\phi}_{12}\bm{\beta}_1} & 
\bm{K}_{\bm{\phi}_{12}\bm{\beta}_2} & 
\bm{K}_{\bm{\phi}_{12}\bm{\phi}_{11}} &
\bm{K}_{\bm{\phi}_{12}\bm{\phi}_{12}} &
\bm{0}_{p_{12},p_{22}}&
\bm{0}_{p_{12},p_{21}}
\\
  \bm{K}_{\bm{\phi}_{22}\bm{\beta}_1} & 
\bm{K}_{\bm{\phi}_{22}\bm{\beta}_2} & 
\bm{0}_{p_{22},p_{11}} &
\bm{0}_{p_{22},p_{12}} &
\bm{K}_{\bm{\phi}_{22}\bm{\phi}_{22}} &
\bm{K}_{\bm{\phi}_{22}\bm{\phi}_{21}} \\
  \bm{K}_{\bm{\phi}_{21}\bm{\beta}_1} & 
\bm{K}_{\bm{\phi}_{21}\bm{\beta}_2} & 
\bm{0}_{p_{21},p_{11}}&
\bm{0}_{p_{21},p_{12}} &
\bm{K}_{\bm{\phi}_{21}\bm{\phi}_{22}} &
\bm{K}_{\bm{\phi}_{21}\bm{\phi}_{21}} \\
		\end{bmatrix},	
}
\end{align*}
where
\begin{align*}
	&\bm{K}_{\bm{\beta}_1\bm{\beta}_1}=
	\varphi_1^{-1}\bm{B}_{11}^\top\bm{W}_1\bm{B}_{11}
	+\varphi_2^{-1}\bm{B}_{21}^\top\bm{W}_2\bm{B}_{21},
	&\qquad
	&
	\bm{K}_{\bm{\beta}_1\bm{\beta}_2}=\bm{K}_{\bm{\beta}_2\bm{\beta}_1}^\top=
	\varphi_1^{-1}\bm{B}_{11}^\top\bm{W}_1\bm{B}_{12}
	+\varphi_2^{-1}\bm{B}_{21}^\top\bm{W}_2\bm{B}_{22},
	\\&
	\bm{K}_{\bm{\beta}_1\bm{\phi}_{11}}=\bm{K}_{\bm{\phi}_{11}\bm{\beta}_1}^\top=
	\varphi_1^{-1}\bm{B}_{11}^\top\bm{W}_1\bm{P}_{11},
	&\qquad
	&
	\bm{K}_{\bm{\beta}_1\bm{\phi}_{12}}=\bm{K}_{\bm{\phi}_{12}\bm{\beta}_1}^\top=
	\varphi_1^{-1}\bm{B}_{11}^\top\bm{W}_1\bm{P}_{12},
	\\&
	\bm{K}_{\bm{\beta}_1\bm{\phi}_{22}}=\bm{K}_{\bm{\phi}_{22}\bm{\beta}_1}^\top=
	\varphi_2^{-1}\bm{B}_{21}^\top\bm{W}_2\bm{P}_{22},
	&\qquad
	&
	\bm{K}_{\bm{\beta}_1\bm{\phi}_{21}}=\bm{K}_{\bm{\phi}_{21}\bm{\beta}_1}^\top=
	\varphi_2^{-1}\bm{B}_{21}^\top\bm{W}_2\bm{P}_{21},
	\\&
	\bm{K}_{\bm{\beta}_2\bm{\beta}_2}=
	\varphi_1^{-1}\bm{B}_{12}^\top\bm{W}_1\bm{B}_{12}
	+\varphi_2^{-1}\bm{B}_{21}^\top\bm{W}_2\bm{B}_{21},
	&\qquad
	&
	\bm{K}_{\bm{\beta}_2\bm{\phi}_{11}}=\bm{K}_{\bm{\phi}_{11}\bm{\beta}_2}^\top=
	\varphi_1^{-1}\bm{B}_{12}^\top\bm{W}_1\bm{P}_{11},
	\\&
	\bm{K}_{\bm{\beta}_2\bm{\phi}_{12}}=\bm{K}_{\bm{\phi}_{12}\bm{\beta}_2}^\top=
	\varphi_1^{-1}\bm{B}_{12}^\top\bm{W}_1\bm{P}_{12},
	&\qquad
	&
	\bm{K}_{\bm{\beta}_2\bm{\phi}_{22}}=\bm{K}_{\bm{\phi}_{22}\bm{\beta}_2}^\top=
	\varphi_2^{-1}\bm{B}_{22}^\top\bm{W}_2\bm{P}_{22},
	\\&
	\bm{K}_{\bm{\beta}_2\bm{\phi}_{21}}=\bm{K}_{\bm{\phi}_{21}\bm{\beta}_2}^\top=
	\varphi_2^{-1}\bm{B}_{22}^\top\bm{W}_2\bm{P}_{21},
	&\qquad
	&
	\bm{K}_{\bm{\phi}_{11}\bm{\phi}_{11}}=
	\varphi_1^{-1}\bm{P}_{11}^\top\bm{W}_1\bm{P}_{11},
	\\&
	\bm{K}_{\bm{\phi}_{11}\bm{\phi}_{12}}=\bm{K}_{\bm{\phi}_{12}\bm{\phi}_{11}}^\top=
	\varphi_1^{-1}\bm{P}_{11}^\top\bm{W}_1\bm{P}_{12},
	&\qquad
	&
	\bm{K}_{\bm{\phi}_{12}\bm{\phi}_{12}}=
	\varphi_1^{-1}\bm{P}_{12}^\top\bm{W}_1\bm{P}_{12},
	\\&
	\bm{K}_{\bm{\phi}_{22}\bm{\phi}_{22}}=
	\varphi_2^{-1}\bm{P}_{22}^\top\bm{W}_2\bm{P}_{22},
	&\qquad
	&
	\bm{K}_{\bm{\phi}_{22}\bm{\phi}_{21}}=
	\bm{K}_{\bm{\phi}_{21}\bm{\phi}_{22}}^\top=
	\varphi_2^{-1}\bm{P}_{22}^\top\bm{W}_2\bm{P}_{21},
	\\&
	\bm{K}_{\bm{\phi}_{21}\bm{\phi}_{21}}=
	\varphi_2^{-1}\bm{P}_{21}^\top\bm{W}_2\bm{P}_{21},
\end{align*}
and $\bm{0}_{m,n}$ a zeros matrix of dimension $m\times n$.
Finally, the matrix
$\bm{K}_{{\varphi\varphi}}$
is expressed as $\bm{K}_{{\varphi\varphi}}=\mathrm{diag}\{K_{\varphi_1\varphi_1}, K_{\varphi_2\varphi_2}\}$,
where
${K}_{{\varphi_1\varphi_{1}}}=-
\bm{d}_1^\top\bm{1}$
and
${K}_{{\varphi_2\varphi_{2}}}=-
\bm{d}_2^\top\bm{1}$,
with
$\bm{d}_k=(d_{k(m+1)},\ldots,d_{kn})^\top$.

\subsection{Confidence intervals and hypothesis testing inference}

We propose to construct confidence intervals for the parameters of the BGAR model considering the asymptotic normality property of the CMLE, $\hat{\bm{\gamma}}$. 
From~\eqref{EQ_normass_CLME}, we have that
\begin{align}\label{EQ_normass_CLME2}
	(\hat{\gamma}_i-{\gamma}_i)
	\left[K(\hat{\bm{\gamma}})^{ii}\right]^{-1/2}
	\sim \mathcal{N}
	(0,1),
\end{align}
approximately, 
where 
$\hat{\gamma}_i$ and ${\gamma}_i$ denote the $i$th component of $\bm{\hat{\gamma}}$ and $\bm{\gamma}$, respectively,
$K(\hat{\bm{\gamma}})^{ii}$ is the $i$th diagonal element of 
$\bm{K}^{-1}(\hat{\bm{\gamma}})$
and
$\mathcal{N}$ denotes a univariate normal distribution. 
Let $\alpha\in (0,1/2)$ be the significance level. A 100($1-\alpha$)\% confidence interval for $\gamma_i$, $i=1,2,\ldots,h$, can be expressed as
\begin{align}\label{EQ_CI}
	\left[
	\hat{\gamma}_i - z_{1-\alpha /2}
	\left(K(\hat{\bm{\gamma}})^{ii}\right)^{1/2};
	\hat{\gamma}_i + z_{1-\alpha /2}
	\left(K(\hat{\bm{\gamma}})^{ii}\right)^{1/2}
	\right],
\end{align}
where $z_{\delta}$ represent the $\delta$ standard normal quantile such that $\Phi(z)=\delta$.

Similarly, from~\eqref{EQ_normass_CLME2}, we can carry out hypothesis testing. Suppose that we wish to test 
$\mathcal{H}_0: \gamma_i= \gamma_i^0$ against $\mathcal{H}_1: \gamma_i\neq\gamma_i^0$, where $\gamma_i^0$ is a given hypothesized value for the true value of  $\gamma_i$. For testing these hypotheses, it can be considered the statistic $z$ 
based on the signed square root of Wald’s statistic, which is computed as~\citep{pawitan2001}
\begin{align*}
	z=\frac{\hat{\gamma}_i-{\gamma}_i^0}{\left(K(\hat{\bm{\gamma}})^{ii}\right)^{1/2}}.
\end{align*}
Under the null hypothesis, $\mathcal{H}_0$, the limiting distribution of $z$ is standard normal.
We reject $\mathcal{H}_0$, whether the absolute observed value of $z$ exceeds the quantile $z_{1-\alpha/2}$ at the $\alpha$ significance level.

\section{Properties of the BGAR model}\label{sec_properties}

In this section, 
we derive stationarity conditions 
for the marginal mean and variance for a BGAR model with identity link function $g_k$. 
The conditions are presented in the following two theorems. These results generalize those derived for univariate GARMA models by~\cite{benjamin2003}.

\begin{theor}
	The marginal mean of 
	$y_{kt}$
	of the BGAR model with identity link function $g_k(\cdot)$ for $k=1,2$ is
	$$\mathrm{I\!E}({y}_{kt}) = \bm{x}_{kt}^\top\bm{\beta}_k,$$
	since the roots of the polynomial $\det(\bm{I}_2 - \bm{A}_1z - \dots - \bm{A}_pz^p)$ lie outside of the unit circle, where  $\bm{I}_2$ is the identity matrix of order $2$,
	$$\bm{A}_{l} = \begin{pmatrix} \phi_{11, l} & \phi_{12, l} \\ \phi_{21, l} & \phi_{22, l} \end{pmatrix}\qquad \text{with}  \qquad l=1,\ldots,p$$
	is the autoregressive coefficients matrix,
	$p = \max(p_{11}, p_{12}, p_{21}, p_{22})$,
	and
	$\phi_{ij,l}=0$
	 when $p_{ij}<p$, for $i, j=1,2$ and $p_{ij}<l\leq p$.
Furthermore, for stationarity, it is required that $\bm{x}_{kt}^\top\bm{\beta}_k = \beta_{k0}$ for all $t$.

\end{theor}

\begin{proof}
	 Considering the identity link function, \eqref{EQ_model_mu} becomes
	 \begin{align}\label{eq_bgar_identity}
	 	&	\mu_{1t}=
	 	\bm{x}_{1t}^\top\bm{\beta}_{1}+
	 	\sum_{i=1}^{p_{11}}\phi_{11,i}
	 	(y_{1t-i}-\bm{x}_{1t-i}^\top\bm{\beta}_{1})+
	 	\sum_{j=1}^{p_{12}}
	 	\phi_{12,j}
	 	(y_{2t-j}-\bm{x}_{2t-j}^\top\bm{\beta}_{2})
	 	\\ \nonumber
	 	&
	 \mu_{2t}=
	 	\bm{x}_{2t}^\top\bm{\beta}_{2}+
	 	\sum_{l=1}^{p_{22}}\phi_{22,l}
	 	(y_{2t-l}-\bm{x}_{2t-l}^\top\bm{\beta}_{2})+
	 	\sum_{s=1}^{p_{21}}
	 	\phi_{21,s}
	 	(y_{1t-s}-\bm{x}_{1t-s}^\top\bm{\beta}_{1}),
	 \end{align}
	
	Let $y_{kt} = \mu_{kt} + v_{kt}$ and $w_{kt} = y_{kt} - x^\top_{kt}\beta_k,$
	where $v_{kt}$ are martigale errors with marginally mean $0$ and uncorrelated~\citep{benjamin2003}. 
	Replacing these quantities in~\eqref{eq_bgar_identity}, we have that
		 \begin{align}\label{eq_bgar_identity2}
		&	w_{1t}=
		\sum_{i=1}^{p_{11}}\phi_{11,i}
		w_{1t-i}+
		\sum_{j=1}^{p_{12}}
		\phi_{12,j}
		w_{2t-j}+v_{1t}
		\\ \nonumber
		&
		w_{2t}=
		\sum_{l=1}^{p_{22}}\phi_{22,l}
		w_{2t-l}+
		\sum_{s=1}^{p_{21}}
		\phi_{21,s}
		w_{1t-s}+v_{2t}.
	\end{align}
Now, we define 	$p = \max(p_{11}, p_{12}, p_{21}, p_{22})$
and
$\phi_{ij,l}=0$
when $p_{ij}<p$, for $i, j=1,2$ and $p_{ij}<l\leq p$.
Then, \eqref{eq_bgar_identity2} becomes
		 \begin{align*}
	&	w_{1t}=
	\sum_{l=1}^{p}(\phi_{11,l}
	w_{1t-l}+
	\phi_{12,l}
	w_{2t-l})+v_{1t}
	\\ \nonumber
	&
	w_{2t}=
	\sum_{l=1}^{p}(\phi_{22,l}
	w_{2t-l}+
	\phi_{21,l}
	w_{1t-l})+v_{2t},
\end{align*}
which can be written in matrix form as
\begin{align}\label{eq_var_p}
	\mathbf{w}_t = \sum_{i=1}^p \bm{A}_i \mathbf{w}_{t-i} + \bm{v}_t, 
\end{align}
	where
	$$\mathbf{w}_t = \begin{pmatrix} w_{1t} \\ w_{2t} \end{pmatrix},
	\quad
\bm{A}_l = \begin{pmatrix} \phi_{11, l} & \phi_{12, l} \\ \phi_{21, l} & \phi_{22, l} \end{pmatrix},
	 \quad
	 \text{and}
	 \quad
	 \bm{v}_t = \begin{pmatrix} v_{1t} \\ v_{2t} \end{pmatrix}.$$
	
	Note that \eqref{eq_var_p} represents a bivariate VAR($p$) model without intercept. The process \eqref{eq_var_p} is stable (and, thus, stationary) if
	$$\det(\bm{I}_2 - \bm{A}_1z - ... - \bm{A}_pz^p) \neq 0 \text{ for } |z| \geq 1,$$
	that is, the roots of the polynomial $\det(\bm{I}_2 - \bm{A}_1z - ... - \bm{A}_pz^p)$ lie outside of the unit circle~\citep{lutkepohl2005new}.

	Since any VAR($p$) process can be written in VAR($1$) form, we can represent \eqref{eq_var_p} as a $2p$-dimensional VAR($1$) given by
	\begin{align*}
		\bm{W}_t = \bm{A}\bm{W}_{t-1} + \bm{V}_t,
	\end{align*}
	where
	$$
	\bm{W}_t = \begin{pmatrix} \bm{w}_t \\ \bm{w}_{t-1} \\  \bm{w}_{t-2} \\ \vdots \\  \bm{w}_{t-p+1} \end{pmatrix}_{2p \times 1},\quad
	\bm{A} = \begin{bmatrix} \bm{A}_1 & \bm{A}_2 & ... & \bm{A}_{p-1} & \bm{A}_p \\ \bm{I}_2 & \bm{0} & ... & \bm{0} & \bm{0} \\ \bm{0} & \bm{I}_2 & ... & \bm{0} & \bm{0} \\ \vdots & \vdots & \ddots & \vdots & \vdots \\ \bm{0} & \bm{0} & ... & \bm{I}_2 & \bm{0} \end{bmatrix}_{2p \times 2p},
	\quad
	\text{and} 
	\quad
	\bm{V}_t = \begin{pmatrix} \bm{v}_t \\ {0} \\ 0 \\ \vdots \\ 0 \end{pmatrix}_{2p \times 1}.$$
	
	Under the stability assumption, the process $\mathbf{W}_t$ has a moving average representation given by
	\begin{align}\label{eq_var_p_ma}
		\bm{W}_t = \sum_{i=0}^\infty \bm{A}^i \bm{V}_{t-i}.
	\end{align}
From~\eqref{eq_var_p_ma}, it follows that 
$\mathrm{I\!E}(\bm{W}_t)=\bm{0}_{2p\times 1}$
	and, therefore, the marginal mean
	$\mathrm{I\!E}({y}_{kt})$ of $y_{kt}$ is given by
	\begin{align*}
	\mathrm{I\!E}({y}_{kt}) = \mathrm{I\!E}(\bm{x}^\top_{kt}\bm{\beta}_k+w_{kt})=
	 \mathrm{I\!E}(\bm{x}^\top_{kt}\bm{\beta}_k)+ \mathrm{I\!E}(w_{kt})=\bm{x}^\top_{kt}\bm{\beta}_k.
	\end{align*}
\end{proof}

\begin{theor}
	The marginal variance-covariance matrix of 
	$\bm{y}_t = (y_{1t}, y_{2t})^\top$
	of the BGAR model with identity link function $g_k(\cdot)$ for $k=1,2$ is 
    $$
    \mathrm{var}(\bm{y}_{t})=\bm{\Phi}
     \mathrm{I\!E}\left(
     \mathcal{A}^{(2)}(B)
V(\bm{\mu}_t)
     \right),
    $$
    where $\bm{\Phi}=\mathrm{diag}\{\varphi_1,\varphi_2\}$,
    $V(\bm{\mu}_t)= \mathrm{diag}\{V(\mu_{1t}),V(\mu_{2t})\}$, and
    $\mathcal{A}^{(2)}(B):=  \sum_{i=0}^\infty \bm{A}^i\bm{A}^{i^\top}$, assuming that the roots of the polynomial $\det(\bm{I}_2 - \bm{A}_1z - ... - \bm{A}_pz^p)$ lie outside of the unit circle.
\end{theor}

\begin{proof}
From proof of Theorem 1, we have that $y_{kt}=\mu_{kt}+v_{kt}$. Hence, we need to compute the marginal variance of the martingale errors, $v_{kt}$, to obtain $\mathrm{var}(y_{kt})$.
Since $\mathrm{I\!E}(v_{kt})=0$, it follows that $\mathrm{var}(v_{kt})=\mathrm{I\!E}(v_{kt}^2)$.

Based on this, we can apply the law of iterated expectations to obtain the marginal variance of $v_{kt}$:
\begin{align*}
    \mathrm{I\!E}(v_{kt}^2)=\mathrm{I\!E}\left(\mathrm{I\!E}\left(v_{kt}^2|\mathcal{F}_{t-1}\right)\right)=\mathrm{I\!E}\left(\mathrm{var}(v_{kt}|\mathcal{F}_{t-1})\right)=\varphi \mathrm{I\!E}(V(\mu_{kt})),
\end{align*}
since $ \mathrm{var}(v_{kt}|\mathcal{F}_{t-1})=\mathrm{var}(y_{kt}|\mathcal{F}_{t-1})$. Thus, 
\begin{align}\label{eq_var_vkt}
\mathrm{var}(v_{kt})=\varphi \mathrm{I\!E}(V(\mu_{kt})).
\end{align}

On the other hand, again from Theorem 1,
\begin{align*}
    w_{kt}=y_{kt} - x^\top_{kt}\beta_k.
\end{align*}
Hence, the marginal variance of $y_{kt}$ can be computed as
\begin{align*}
    \mathrm{var}(y_{kt})=  \mathrm{var}(x^\top_{kt}\beta_k+w_{kt})= \mathrm{var}(w_{kt})=\mathrm{I\!E}(w_{kt}^2),
\end{align*}
since $\mathrm{I\!E}(w_{kt})=0$, because $\mathrm{I\!E}(y_{kt})= x^\top_{kt}\beta_k$ by Theorem 1.

Supposing the roots of the polynomial $\det(\bm{I}_2 - \bm{A}_1z - ... - \bm{A}_pz^p)$ lie outside of the unit circle so that the process admits an infinite moving average representation, we have that
\begin{align}\label{eq_cov_wt}
     \mathrm{var}(\bm{W}_t)=\mathrm{I\!E}(\bm{W}_t\,\bm{W}_t^\top)=
\mathrm{I\!E}\left(
    \sum_{i=0}^\infty \bm{A}^i \bm{V}_{t-i} \sum_{j=0}^\infty \bm{V}_{t-j}^\top\bm{A}^{j\top}
     \right)
    =
    \sum_{i=0}^\infty \sum_{j=0}^\infty \bm{A}^i
     \mathrm{I\!E}\left(
     \bm{V}_{t-i} \bm{V}_{t-j}^\top
     \right)\bm{A}^{j\top}.
\end{align}
Note that the martigale errors, $\bm{v}_t$, are uncorrelated by definition. That is,
\begin{align*}
 \mathrm{I\!E}(\bm{v}_{t}\bm{v}_{s}^\top)=
 \begin{cases}
  0  & \text{if } t\neq s\\
  \bm{\Phi}\mathrm{I\!E}(V(\bm{\mu}_t)) & \text{if } t = s,
\end{cases}
\end{align*}
where $\bm{\Phi}=\mathrm{diag}\{\varphi_1,\varphi_2\}$ and $V(\bm{\mu}_t)=\mathrm{diag}\{V(\mu_{1t}),V(\mu_{2t})\}$.
Hence, the only non-zero contribution in the double summation in \eqref{eq_cov_wt} occurs when $i=j$:
\begin{align*}
     \mathrm{var}(\bm{W}_t)=
    \sum_{i=0}^\infty \bm{A}^i
     \mathrm{I\!E}\left(
     \bm{V}_{t-i} \bm{V}_{t-i}^\top
     \right)\bm{A}^{i\top},
\end{align*}
where
\begin{align*}
     \mathrm{I\!E}\left( \bm{V}_{t-i} \bm{V}_{t-i}^\top
     \right)
     =\begin{bmatrix}
         \bm{\Phi}\mathrm{I\!E}(V(\bm{\mu}_{t-i})) & \bm{0} & \ldots & \bm{0} \\
         \bm{0} & \bm{0} &  \ldots  & \bm{0} \\
         \vdots & \vdots & \ddots & \vdots\\
          \bm{0} & \bm{0} &   \ldots  & \bm{0} 
     \end{bmatrix}_{2p \times 2p}.
\end{align*}
Let
$\mathcal{A}^{(2)}(B):=  \sum_{i=0}^\infty \bm{A}^i\bm{A}^{i^\top}B^i$. It follows that
\begin{align*}
     \mathrm{var}(\bm{W}_t)=
       \bm{\Phi}
     \mathrm{I\!E}\left(
     \mathcal{A}^{(2)}(B)
V(\bm{\mu}_{t})
     \right).
\end{align*}
Then, the marginal variance-covariance of $\boldsymbol{y}_t$ is given by:
\begin{align}\label{eq_var_yt}
     \mathrm{var}(\bm{y}_t)=   \bm{\Phi}
     \mathrm{I\!E}\left(
     \mathcal{A}^{(2)}(B)
V(\bm{\mu}_t)
     \right).
\end{align}
\end{proof}

\begin{lemma}\label{lemma}
    For any BGAR$(p_{11},p_{12},p_{22},p_{21})$ model with identity link function, the marginal variance-covariance matrix of $\bm{\mu}_t$ is expressed as
    \begin{align}\label{eq_res_lemma}
        \mathrm{var}(\bm{\mu}_t)=
        \bm{\Phi}\mathrm{I\!E}\left(
          (\mathcal{A}^{(2)}(B)-\bm{I}_2)
V(\bm{\mu}_t)
        \right).
    \end{align}
\end{lemma}

\begin{proof}[Proof of Lemma 1]
    We define $y_{kt} = \mu_{kt} + v_{kt}$, for $k=1,2$. That is, $\bm{y}_t=\bm{\mu}_t+\bm{v}_t$. Since $\bm{v}_t$ is uncorrelated with $\bm{\mu}_t$, it follows that
    \begin{align*}
          \mathrm{var}(\bm{y}_t)= \mathrm{var}(\bm{\mu}_t)+ \mathrm{var}(\bm{v}_t).
    \end{align*}
    Thus, using~\eqref{eq_var_yt} and~\eqref{eq_var_vkt}, we obtain the result~\eqref{eq_res_lemma} of Lemma~\ref{lemma}.
\end{proof}

\begin{corollary}
   The stationary marginal variance of $\bm{y}_t$ for the Poisson-Poisson BGAR model considering identity link functions for each mean is
   \begin{align*}
        \mathrm{var}(\bm{y}_t)=  
     \mathcal{A}^{(2)}(1)\bm{\beta}_0,
   \end{align*}
   where
  $
         \mathcal{A}^{(2)}(1)= \sum_{i=0}^\infty\bm{A}^i \bm{A}^{i\top}
  $.
\end{corollary}

\begin{proof}[Proof of Corollary 1]
For the Poisson distribution, we have that $V(\bm{\mu}_t)=\bm{\mu}_t$ and $\bm{\Phi}=\bm{I}_2$.
From~\eqref{eq_var_yt}, we obtain 
\begin{align*}
 \mathrm{var}(\bm{y}_t)=  \bm{I}_2
 \mathcal{A}^{(2)}(B)
     \mathrm{I\!E}\left(
\bm{\mu}_t
     \right)=
     \mathcal{A}^{(2)}(1)\bm{\beta}_0.
\end{align*}
\end{proof}

Note that for other combinations of distributions, the derivation of the marginal variance is analogous to that in Corollary~1. It is only necessary to replace the variance function $V(\bm{\mu}_t)$ appropriately according to each distribution and perform some algebraic manipulations.

\section{Model selection, diagnostic analysis and forecasting}\label{sec_diag_analysis}

This section presents techniques for model selection based on information criteria, which can be used for model identification and some diagnostic analysis tools for the fitted BGAR models. Diagnostic checks are helpful to identify whether the model suitably captures the data dynamics. We also present how to obtain in-sample and out-of-sample forecasting from the proposed model.

\subsection{Model selection}

The BGAR($p_{11},p_{12},p_{22},p_{21}$) model selection requires identifying the autoregressive orders $p_{ij}$, for $i,j=1,2$, that provide the best fit to the data at hand. For this, we need to combine different orders $p_{ij}$ and evaluate what is the best-fitted model according to some criterion. We adopt two  widely used model selection criteria based on the log-likelihood function, such as Akaike’s information criterion (AIC)~\citep{akaike1973} and Bayesian information 
criterion (BIC)~\citep{akaike1973_bic}.
Let $\hat{\ell}$ be the conditional 
log-likelihood function evaluated at the CMLEs, $\ell(\hat{\bm \gamma})$. The AIC and BIC criteria are defined as
\begin{align*}
	\text{AIC}(v) = 2v-2\hat{\ell}
\end{align*}
and
\begin{align*}
	\text{BIC}(v) = v\log (n)-2\hat{\ell},
\end{align*}
respectively,
where
$v$ is the number of estimated parameters. For the Binomial, Negative Binomial, and Poisson distributions, replace $v$ by $v^\star$.
The model with the lower AIC and BIC is selected from a finite set of competing models.

\subsection{Residual analysis}

Residual analysis plays an important role in the goodness-of-fit identification of a fitted model~\citep{fokianos2004partial}. 
For the proposed BGAR model, we consider the quantile residuals and randomized quantile residuals~\citep{dunn1996}, which are general tools widely used for diagnosing lack of fit in several classes of models.

Let $F_k(y_{kt};\bm{\gamma}|\mathcal{F}_{t-1})$ be the conditional cumulative distribution function of conditional density~\eqref{EQ_pdf_expfam}. In the case in which $F$ is continuous $F_k(y_{kt};\bm{\gamma}|\mathcal{F}_{t-1})$ has a uniform distribution on the unit interval, and the quantile residuals are 
defined as
\begin{align*}
	r_{kt}^q=\Phi^{-1}\left[F_k(y_{kt};\hat{\bm{\gamma}}|\mathcal{F}_{t-1})\right],
\end{align*}
where $\Phi^{-1}(\cdot)$ is the quantile function of the standard normal distribution.

A randomization is needed to produce continuously distributed residuals when the response variable is discrete. In this case, if $F$ is not continuous, the randomized quantile residuals for $y_{kt}$ are given by
\begin{align*}
	r_{kt}^{rq}=\Phi^{-1}(u_{kt}),
\end{align*}
where $u_{kt}$ is a uniform random variable on the interval $(a_{kt},b_{kt}]$, being 
$a_{kt}=F_k(y_{kt}-1;\hat{\bm{\gamma}}|\mathcal{F}_{t-1})$
and
$b_{kt}=F_k(y_{kt};\hat{\bm{\gamma}}|\mathcal{F}_{t-1})$. 
If $\hat{\bm{\gamma}}$ is consistently estimated, the residuals $r_{kt}^{rq}$ and $r_{kt}^{q}$ converge to the standard normal \citep{dunn1996}. 
This means that under the correct model specification, the residuals are serially uncorrelated and are normally distributed with zero mean and unit variance.

We also consider composite residuals analogous to those used for model diagnostics in Dirichlet regression~\citep{gueorguieva2008dirichlet}. The composite residuals for the BGAR model are defined as follows:
\begin{align*}
    r_t = (r_{1t}^{(\cdot)})^2 + (r_{2t}^{(\cdot)})^2,
\end{align*}
which combines the squared residuals from both series into a single measure. Under the assumption of independence and standard normality of the individual residuals, \( r_t \) follows a chi-squared distribution with two degrees of freedom. 

Residual analysis can be performed using graphical inspections. First, QQ-plots are useful for checking the normality assumption of the quantile residuals. 
Second,  if the model captured all the serial correlation, the residual autocorrelation function (ACF) plot for $r_{kt}^{(\cdot)}$, $k=1,2$, must indicate that all autocorrelations are statistically non-significant.
Third, a plot of \( r_t \) against the time index may be used to verify the model fit, where most points should lie below the 95th quantile of the chi-squared distribution. The absence of systematic patterns or clustering in this plot indicates that the model adequately captures the joint dynamics of the series.
Finally, significant correlations at any lag in the residual cross-correlation function (CCF)~\citep{box2015time} plot may indicate that the dependence between the series has not been successfully captured or that a more complex model structure is needed. Thus, the model is appropriate only if there is no significant residual cross-correlation. 

\subsection{Forecasting}\label{sec_predict_forec}

Forecasting from the BGAR($p_{11},p_{12},p_{22},p_{21}$) model also can be determined using the similar theory of GARMA
models of univariate time series analysis.
The predicted values of $\bm{y}_k$ (in-sample forecasting) are denoted by $\hat{\bm{\mu}}_k$, and are obtained from~\eqref{EQ_model_mu} evaluated on the CMLE $\hat{\bm{\gamma}}$. The predicted values from a fitted BGAR($p_{11},p_{12},p_{22},p_{21}$) model are expressed as
\begin{align*}
	&	\hat{\mu}_{1t}=
	g_1^{-1}\left(
	\bm{x}_{1t}^\top\bm{\hat{\beta}}_{1}+
	\sum_{i=1}^{p_{11}}\hat{\phi}_{11,i}
	\left[g_1(y_{1t-i})-\bm{x}_{1t-i}^\top\bm{\hat{\beta}}_{1}\right]+
	\sum_{j=1}^{p_{12}}
	\hat{\phi}_{12,j}
	\left[g_2(y_{2t-j})-\bm{x}_{2t-j}^\top\bm{\hat{\beta}}_{2}\right]
	\right)
	\\ \nonumber
	&
\hat{\mu}_{2t}=	g_2^{-1}\left(
	\bm{x}_{2t}^\top\bm{\hat{\beta}}_{2}+
	\sum_{l=1}^{p_{22}}\hat{\phi}_{22,l}
	\left[g_2(y_{2t-l})-\bm{x}_{2t-l}^\top\bm{\hat{\beta}}_{2}\right]+
	\sum_{s=1}^{p_{21}}
\hat{\phi}_{21,s}
	\left[g_1(y_{1t-s})-\bm{x}_{1t-s}^\top\bm{\hat{\beta}}_{1}\right]
	\right),
\end{align*}
where $t=m+1,\ldots,n$.

Since a fitted model passes in all checks diagnostic analysis, it can be used to obtain out-of-sample forecasts for the time series $\bm{y}_k$ 
 at hand. 
Let $h_0$ be the desirable forecast horizon.
When we observed the series until time $n$, the forecasted values of $\hat{\mu}_{k(n+h)}$, for $h=1,2,\ldots,h_0$, are sequentially computed as
\begin{align*}
	&	\hat{\mu}_{1(n+h)}=
	g_1^{-1}\left(
	\bm{x}_{1(n+h)}^\top\bm{\hat{\beta}_{1}}+
	\sum_{i=1}^{p_{11}}\hat{\phi}_{11,i}
	\left[g_1^\star(y_{1(n+h-i)})-\bm{x}_{1(n+h-i)}^\top\bm{\hat{\beta}_{1}}\right]+
	\sum_{j=1}^{p_{12}}
	\hat{\phi}_{12,j}
	\left[g_2^\star(y_{2(n+h-j)})-\bm{x}_{2(n+h-j)}^\top\bm{\hat{\beta}_{2}}\right]
	\right)
	\\ \nonumber
	&
	\hat{\mu}_{2(n+h)}=	g_2^{-1}\left(
	\bm{x}_{2(n+h)}^\top\bm{\hat{\beta}_{2}}+
	\sum_{l=1}^{p_{22}}\hat{\phi}_{22,l}
	\left[g_2^\star(y_{2(n+h-l)})-\bm{x}_{2(n+h-l)}^\top\bm{\hat{\beta}_{2}}\right]+
	\sum_{s=1}^{p_{21}}
	\hat{\phi}_{21,s}
	\left[g_1^\star(y_{1(n+h-s)})-\bm{x}_{1(n+h-s)}^\top\bm{\hat{\beta}_{1}}\right]
	\right),
\end{align*}
where 
\begin{align*}
	g_k^\star(y_{kt})=\left\{
	\begin{array}{ll}
		g_k(\hat{\mu}_{kt}) &\, \text{if} \,\,\, t >n,\\
		g_k(y_{kt})  & \, \text{if} \,\,\, t \leq n,\\
	\end{array}
	\right.
\end{align*}
for $k=1,2$. The covariates $x_{kt}$ are easily computed for $t>n$ since we suppose they are deterministic functions of $t$, such as sines and cosines in harmonic analysis, dummy variables, or polynomial trends.

Both the in-sample and out-of-sample forecasts can be assessed by considering the same measures of forecast accuracy typically applied
to univariate time series data. 
To assess out-of-sample forecasts, the last $h_0$ observations were removed of the sample and the BGAR($p_{11},p_{12},p_{22},p_{21}$) model was fitted with the $n-h_0$  remaining observations. 
We consider the following measures:
root mean squared error (RMSE),
mean absolute error (MAE), 
and
mean absolute percentage error (MAPE)
~\citep{hyndman2006,hyndman2018forecasting},
that are computed as
\begin{align*}
	\text{RMSE}_k = \sqrt{\frac{1}{N}\sum_{t=n+1}^N e_{kt}^2, }\quad 
	\quad
	\text{MAE}_k = \frac{1}{N}\sum_{t=n+1}^N |e_{kt}|, 
	\quad
\quad\text{and} 
\quad
	\text{MAPE}_k = \frac{1}{N}\sum_{t=n+1}^N|p_{kt}|, 
\end{align*}
where 
$e_{kt}:=y_{kt} - \hat{\mu}_{kt}$ is the forecast error,
$p_{kt}:=100\times {e_{kt}}/{y_{kt}}$ is the  percentage error, 
for $k = 1,2$, and
 $N=n+h_0$.
It is worth noting that, for the discrete distributions, the MAE$_k$ and MAPE$_k$ measures can be used only since $y_{kt}\neq 0$ for all $t$.

These quantities are useful for empirical comparisons of competing models and for assessing the forecast performance of the proposed model. 
In particular, the MAPE is independent of the scale of the data. Thus, it
can be used to compare the forecast accuracy of time series with different units~\citep{hyndman2006}.
All these measures can be easily obtained in \verb|R| language using the \verb|accuracy| function from \verb|forecast|~package~\citep{hyndman2015forecasting,hyndman2008automatic}. The model that provides the forecasts with more accuracy presents the smallest RMSE, MAE, 
and
MAPE
values.

\section{Numerical evaluation}\label{sec_simu}

We carry out Monte Carlo simulation studies to evaluate the finite sample performance of the CMLE. 
We consider $R=10,000$ Monte Carlo replications and sample sizes $n\in\{100,250,500\}$.  
In the models with random component $y_{kt}$ ($k=1,2$) following a discrete conditional distribution, we establish a threshold of $c=0.1$ in the simulation study and real data application in Section~\ref{sec_app}. 
This means, in~\eqref{EQ_model_mu}, we replace $y_{kt-i}$ (and similarly $y_{kt-j}$,  $y_{kt-l}$, $y_{kt-s}$) by $c$ if $y_{kt-i}=0$.
We compute the mean of all estimates,
percentage relative bias (RB\%), and mean square error (MSE) to evaluate the point estimates' performance. Moreover, the coverage rates (CR) from confidence intervals are reported to assess the performance of interval estimation. The CR is determined by the proportion of Monte Carlo replicas in which the parameter is within the confidence interval.
The limits of the confidence intervals are computed
from~\eqref{EQ_CI}, considering a significance level of $\alpha=5\%$.
The optimization of conditional log-likelihood function was performed using the BFGS quasi-Newton method, incorporating analytical first derivatives.
Starting values for the autoregressive parameters
were obtained as described at the end of Section~\ref{sec_score_vector}. All the simulations were carried out using the \verb|R|
statistical computing environment~\citep{R_ref}.

We adopted settings similar to the application study (see Section~\ref{sec_app}) and several other scenarios with different combinations of conditional distributions for the random component. Specifically, we considered a  BGAR model with Negative Binomial (NB) conditional distributions for each series,  denoted as $\bm{y}_{k}=(y_{1t}, y_{2t})^\top$,  and logarithmic link function for $g_k(\cdot)$. In each Monte Carlo replication, we simulated two vectors of length $n$, representing observations of the variables $y_{kt}$, from the following BGAR models:
\begin{itemize}
	\item NB-NB-BGAR($1,1,1,1$) model
	with a seasonal covariate: it includes an intercept and a cosine term, $\bm{x_{kt}}=(\mathrm{cos}(2\pi t/12))$ to capture seasonality with a period of $12$ time units. We consider the precision parameters $\kappa_1=12$ and  $\kappa_2=20$.
	\item NB-NB-BGAR$(1,1,1,1)$ model without covariates: it does not include any exogenous variables, and the precision parameters are $\kappa_1=12$ and  $\kappa_2=10$.
\end{itemize}

Table~\ref{tab_nbnbgar1111_cov} presents the selected parameter values for these models and the simulation results from the two scenarios. Including covariates in the NB-NB-BGAR model generally leads to smaller RB\% values (closer to zero) than the model without covariates, suggesting that covariates improve the accuracy of parameter estimates. As the sample size increases, both models exhibit a decrease in RB\% and MSE values for the most parameter estimates. This is expected since the CMLE's bias and variance decrease asymptotically as the sample size increases. Finally, regarding interval estimation, we observe that the CR generally approaches $95\%$ for both models as the sample size increases, indicating that the CR converges to nominal level.

\begin{table}[!htbp] 
	\small
	\centering 
	\caption{
    Simulation results from NB-NB-BGAR($1,1,1,1$) model with and without covariates.
     } 
	\label{tab_nbnbgar1111_cov} 
	\begin{tabular}{@{\extracolsep{5pt}} lrrrrrrrr} 
		\\[-1.8ex]
		\hline \\[-1.8ex] 
		\multicolumn{9}{c}{Parameter values}\\
		\hline\\[-1.8ex] 
		Measures	& $\beta_{10}$ & $\beta_{11}$  & $\beta_{20}$ & $\beta_{21}$  & $\phi_{11}$ & $\phi_{12}$ & $\phi_{22}$ & $\phi_{21}$ \\ 	
		& $3.5$ & $1.4$  & $3.0$ & $0.7$  & $0.3$ & $-0.1$ & $0.2$ & $0.2$ \\ 
        		\hline\\[-1.8ex] 
		\multicolumn{9}{c}{With covariates}\\
		\hline\\[-1.8ex] 
		\multicolumn{9}{c}{$n=100$}\\
		\hline\\[-1.8ex] 
	Mean   & $3.4953$ & $1.4014$ & $2.9949$ & $0.7019$ & $0.2733$ & $-0.1038$ & $0.1768$ & $0.2056$ \\
RB\%   & $-0.1348$ & $0.0985$ & $-0.1688$ & $0.2686$ & $-8.9057$ & $3.8219$ & $-11.6095$ & $2.8144$ \\
MSE    & $0.0028$ & $0.0045$ & $0.0019$ & $0.0032$ & $0.0092$ & $0.0103$ & $0.0088$ & $0.0076$ \\
CR     & $0.9386$ & $0.9421$ & $0.9444$ & $0.9451$ & $0.9431$ & $0.9524$ & $0.9499$ & $0.9503$ \\
\hline\\[-1.8ex]
\multicolumn{9}{c}{$n=250$} \\
\hline\\[-1.8ex]
Mean   & $3.4993$ & $1.3994$ & $2.9984$ & $0.7002$ & $0.2890$ & $-0.1015$ & $0.1909$ & $0.2015$ \\
RB\%   & $-0.0189$ & $-0.0395$ & $-0.0526$ & $0.0256$ & $-3.6556$ & $1.4999$ & $-4.5409$ & $0.7486$ \\
MSE    & $0.0011$ & $0.0018$ & $0.0008$ & $0.0012$ & $0.0033$ & $0.0039$ & $0.0033$ & $0.0028$ \\
CR     & $0.9538$ & $0.9492$ & $0.9510$ & $0.9550$ & $0.9532$ & $0.9529$ & $0.9527$ & $0.9557$ \\
\hline\\[-1.8ex]
\multicolumn{9}{c}{$n=500$} \\
\hline\\[-1.8ex]
Mean   & $3.4997$ & $1.3990$ & $2.9995$ & $0.6998$ & $0.2931$ & $-0.1015$ & $0.1946$ & $0.2003$ \\
RB\%   & $-0.0095$ & $-0.0681$ & $-0.0173$ & $-0.0319$ & $-2.3112$ & $1.5088$ & $-2.6788$ & $0.1535$ \\
MSE    & $0.0005$ & $0.0009$ & $0.0004$ & $0.0006$ & $0.0017$ & $0.0019$ & $0.0017$ & $0.0014$ \\
CR     & $0.9550$ & $0.9553$ & $0.9513$ & $0.9581$ & $0.9519$ & $0.9522$ & $0.9525$ & $0.9570$ \\
		\hline\\[-1.8ex] 
		\multicolumn{9}{c}{Without covariates}\\
		\hline\\[-1.8ex] 
		\multicolumn{9}{c}{$n=100$}\\
\hline\\[-1.8ex]
Mean   & $3.4966$ & $-$ & $2.9949$ & $-$ & $0.2765$ & $-0.1050$ & $0.1814$ & $0.2043$ \\
RB\%   & $-0.0970$ & $-$ & $-0.1694$ & $-$ & $-7.8326$ & $5.0348$ & $-9.3117$ & $2.1693$ \\
MSE    & $0.0024$ & $-$ & $0.0026$ & $-$ & $0.0087$ & $0.0064$ & $0.0084$ & $0.0119$ \\
CR     & $0.9372$ & $-$ & $0.9420$ & $-$ & $0.9483$ & $0.9532$ & $0.9520$ & $0.9524$ \\
\hline\\[-1.8ex]
\multicolumn{9}{c}{$n=250$} \\
\hline\\[-1.8ex]
Mean   & $3.4994$ & $-$ & $2.9984$ & $-$ & $0.2922$ & $-0.1012$ & $0.1937$ & $0.2018$ \\
RB\%   & $-0.0158$ & $-$ & $-0.0534$ & $-$ & $-2.5948$ & $1.1692$ & $-3.1634$ & $0.8919$ \\
MSE    & $0.0010$ & $-$ & $0.0010$ & $-$ & $0.0034$ & $0.0024$ & $0.0033$ & $0.0045$ \\
CR     & $0.9445$ & $-$ & $0.9502$ & $-$ & $0.9489$ & $0.9493$ & $0.9525$ & $0.9538$ \\
\hline\\[-1.8ex]
\multicolumn{9}{c}{$n=500$} \\
\hline\\[-1.8ex]
Mean   & $3.4994$ & $-$ & $2.9992$ & $-$ & $0.2962$ & $-0.1012$ & $0.1965$ & $0.2003$ \\
RB\%   & $-0.0163$ & $-$ & $-0.0269$ & $-$ & $-1.2687$ & $1.2467$ & $-1.7526$ & $0.1444$ \\
MSE    & $0.0005$ & $-$ & $0.0005$ & $-$ & $0.0016$ & $0.0012$ & $0.0017$ & $0.0023$ \\
CR     & $0.9469$ & $-$ & $0.9512$ & $-$ & $0.9536$ & $0.9501$ & $0.9525$ & $0.9499$ \\
\hline\\[-1.8ex]
	\end{tabular} 
\end{table}

To define the other schemes and combinations of conditional distributions for each series, we use the following abbreviations for the distributions:  Poisson (P), normal (N), gamma (G), and inverse normal (IN). Table~\ref{tab_others_mc_simu} displays the other five considered scenarios. It presents each model's conditional distributions, link functions, and parameter values. 

\begin{table}[!htbp] 
	\scriptsize
	\centering 
	\caption{Other simulation scenarios considering different models and combinations of parameters.} 
	\label{tab_others_mc_simu} 
	\begin{tabular}{@{\extracolsep{5pt}} lllllll} 
		\\[-1.8ex]
		\hline \\[-1.8ex] 
		Model 
		 & Distribution & Link & Intercept & \multicolumn{2}{c}{Autoregressive terms} & $\varphi_k$ \\ 
		\hline \\[-1.8ex] 
	\multirow{2}{*}{G-G-BGAR($1,1,1,1$)}&   Gamma  & $\log$  &  $\beta_{10}=2.8$  & $\phi_{11}=0.3$ & $\phi_{12}=0.3$ &  $\varphi_{1}=0.2$  \\
		& Gamma  & $\log$ &  $\beta_{20}=3.1$  & $\phi_{22}=0.1$ & $\phi_{21}=-0.2$  &$\varphi_{2}=0.4$   \\
		\hline \\[-1.8ex] 
		\multirow{2}{*}{N-N-BGAR($2,1,2,1$)}&   Normal  & id. &  $\beta_{10}=0.7$  & $\bm{\phi_{11}}=(-0.6,0.2)^\top$ & $\phi_{12}=-0.2$  & $\varphi_{1}=14$ \\
		&  Normal  & id. &  $\beta_{20}=0.5$  & $\bm{\phi_{22}}=(0.5,-0.1)^\top$ & $\phi_{21}=0.1$ &$\varphi_{2}=16$   \\
		\hline\\[-1.8ex] 
		\multirow{2}{*}{N-G-BGAR($1,0,1,1$)}&   Normal  & id. &  $\beta_{10}=2.5$  & $\phi_{11}=0.2$ & $-$  & $\varphi_{1}=5$ \\
		&  Gamma  & $\log$ &  $\beta_{20}=2.1$  & $\phi_{22}=0.3$ & $\phi_{21}=0.2$& $\varphi_{2}=0.5$   \\
		\hline \\[-1.8ex] 
		\multirow{2}{*}{P-IN-BGAR($1,1,1,1$)}&  Poisson  & $\log$  &  $\beta_{10}=1.8$  & $\phi_{11}=0.3$ & $\phi_{12}=0.1$ &$-$    \\
		&   I. Normal  & $\log$ &  $\beta_{20}=0.1$  & $\phi_{22}=0.2$ & $\phi_{21}=0.2$ & $\varphi_{2}=0.5$   \\
		\hline \\[-1.8ex] 
		\multirow{2}{*}{P-N-BGAR($2,2,1,0$)}&   Poisson  & $\log$  &  $\beta_{10}=3.1$  & $\bm{\phi_{11}}=(0.3,-0.1)^\top$ &  $\bm{\phi_{12}}=(-0.2,0.1)^\top$ & $-$  \\
		&   Normal  & id. &  $\beta_{20}=1.5$  & $\phi_{22}=0.3$ & $-$ & $\varphi_{2}=9$    \\
		\hline \\[-1.8ex] 
	\end{tabular} 
\end{table}

\begin{figure}[!htpb]
	\centering
	\includegraphics[width=1.02\textwidth]{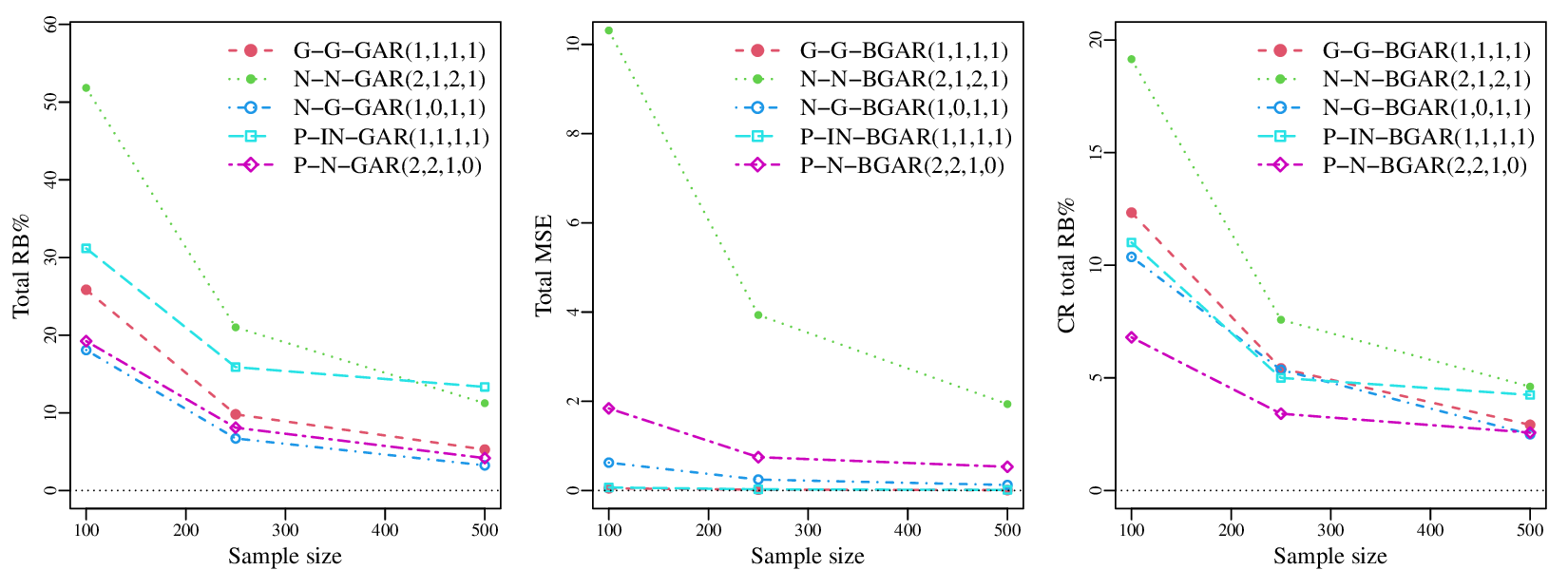}
	\caption{Monte Carlo simulation results from other considered scenarios.}
	\label{fig_res_other_mc_simu}
\end{figure}

For Table~\ref{tab_others_mc_simu} scenarios, we compute the sum of the RB\% and MSE of the estimates of all the parameters for each model and sample size, defined as total RB\% (in absolute value) and MSE. Similarly, we compute the total RB$\%$ of the CR expected of $0.95$ from confidence intervals of each model's parameters (CR total RB$\%$). It is defined as the sum of the absolute relative biases of the computed CR values, with respect to the expected CR of $0.95$, across all parameters.
This metric quantifies the overall deviation of the empirical CR from the expected 95\% level.
Figure~\ref{fig_res_other_mc_simu} presents the simulation results considering these measures. We observe that the total RB\% and MSE decrease for all considered models as the sample size increases. Besides, the CR total RB$\%$ tend to zero as sample size increases in all schemes.

\section{An application to cases and hospitalizations due to Leptospirosis}\label{sec_app}

Leptospirosis is a zoonotic bacterial disease caused by spirochete bacteria from the genus Leptospira~\citep{terpstra2003human}.
It represents a global health concern, especially in regions with warm and humid climates~\citep{bharti2003leptospirosis}. The disease is transmitted to humans through contact with water or soil contaminated by the urine of infected animals~\citep{bharti2003leptospirosis}. It can manifest with a wide range of symptoms, including fever, muscle aches, jaundice, and potential kidney or liver failure. Early diagnosis and prompt antibiotic treatment are essential for a favorable prognosis.

In São Paulo, Brazil, leptospirosis presents a substantial public health challenge, being considered an endemic disease. The tropical climate and urbanization of the state create an environment conducive to the spread of the disease, resulting in a higher risk of exposure to contaminated water sources. São Paulo has experienced outbreaks of leptospirosis, particularly during periods of heavy rainfall and flooding, when the risk of exposure to contaminated water is increased~\citep{blanco2015fifteen, de2022epidemiologia}. 
Understanding its prevalence, transmission dynamics, and implementing preventive measures are important steps to safeguard the health of the local population.

In this application, we analyze the time series of the monthly number of leptospirosis confirmed cases ($y_1$) and hospitalizations due to leptospirosis  ($y_2$) in São Paulo state, Brazil, from January 2001 to December 2022 (totalizing $264$ observations). 
The data are publicly available from the Department of Informatics of the Brazilian Unified Health System (DATASUS) at~\citep{DATASUS}.
Our primary interest is to consider the leptospirosis cases series to forecast hospitalizations.
To assess the out-of-sample forecasting performance of the proposed model, we set aside the last $12$ months from each time series and fit the model using the remaining $252$ observations.

Table~\ref{tab_desc_stat_leptospirosis} presents descriptive statistics of the leptospirosis cases and hospitalizations series. 
 The number of cases series has a median of $42.00$ and a mean of $58.15$, while the hospitalizations series has a median of $27.00$ and a mean of $35.15$. Furthermore, the variances and quartile values indicate the data spread, with the hospitalization time series showing a narrower interquartile range than the number of cases time series.

\begin{table}[!htbp] \centering 
	\caption{Descriptive statistics of monthly cases and hospitalizations ($n=264$).} 
	\label{tab_desc_stat_leptospirosis} 
	\begin{tabular}{@{\extracolsep{5pt}} l|rrrrrrr} 
	\hline 
	 \\[-1.8ex] 
	Response &	Min. & 1st Qu. & Median & Mean & 3rd Qu. & Max. & Variance \\ 
		\hline \\[-1.8ex]  
Cases &	 $7.00$ & $28.00$ & $42.00$ & $57.31$ & $75.50$ & $279.00$ & $1,682.69$ \\ 
Hospitalizations& $4.00$ & $19.00$ & $26.50$ & $34.66$ & $47.00$ & $164.00$ & $495.49$ \\ 
		\hline 
	\end{tabular} 
\end{table}

Figure~\ref{fig_y1_stl_plot} illustrates the seasonal decomposition of leptospirosis cases series using loess (STL), showcasing its seasonal, trend, and residual components; see \citep{cleveland1990stl}.
We can observe a monthly seasonality effect by examining the seasonal subseries plot of leptospirosis cases. 
According to the monthplot in Figure~\ref{fig_y1_monthplot}, the number of cases is higher from January to March, decreases from April, and then increases again from September.
This increase occurs in the rainy season in São Paulo.
%
Figures~\ref{fig_y2_stl_plot} and \ref{fig_y2_monthplot}
 display the STL and seasonal subseries
 for the hospitalizations series. According to Figure~\ref{fig_y2_monthplot}, this series also exhibits monthly seasonality, with hospitalizations by leptospirosis increasing from January to March, decreasing from April to October, and then rising again from November.

\begin{figure}[!htpb]
	\centering
	\subfigure[ Seasonal trend	decomposition of the number of observed cases of leptospirosis.]{
		\includegraphics[width=.45\linewidth]{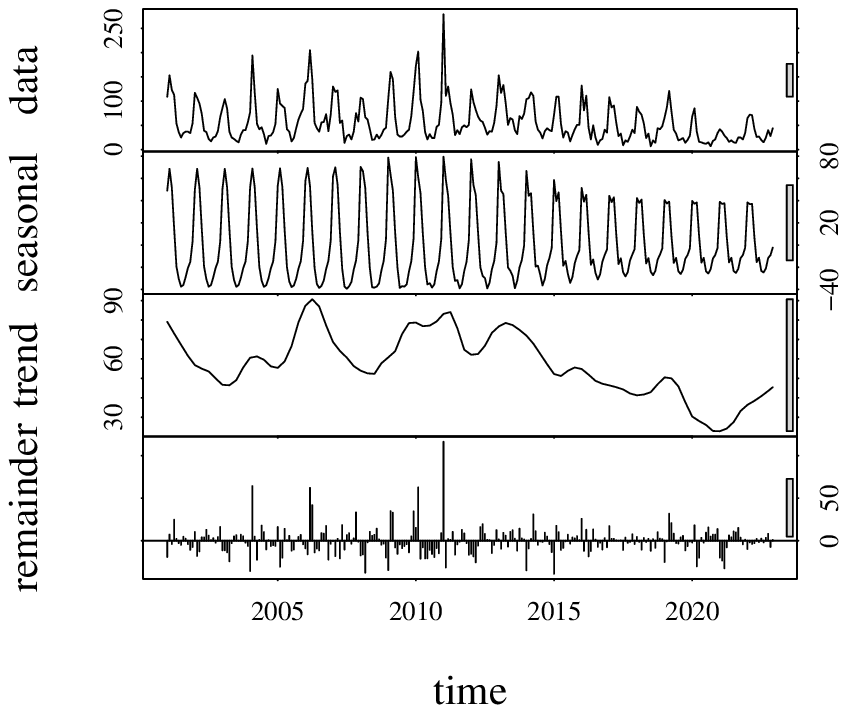}\label{fig_y1_stl_plot}
	}
	\subfigure[Monthplot of the number of observed cases of leptospirosis.]{\includegraphics[width=.41\linewidth]{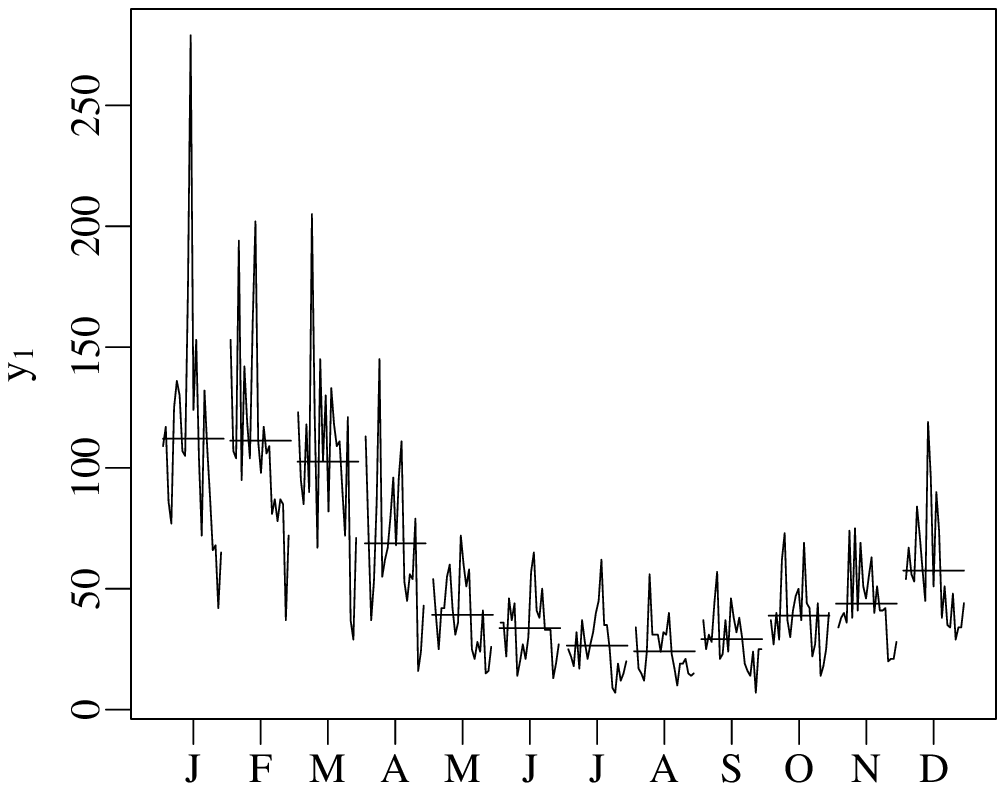}
		\label{fig_y1_monthplot}}
		\subfigure[ Seasonal trend	decomposition of the number of hospitalizations by leptospirosis.]{
		\includegraphics[width=.45\linewidth]{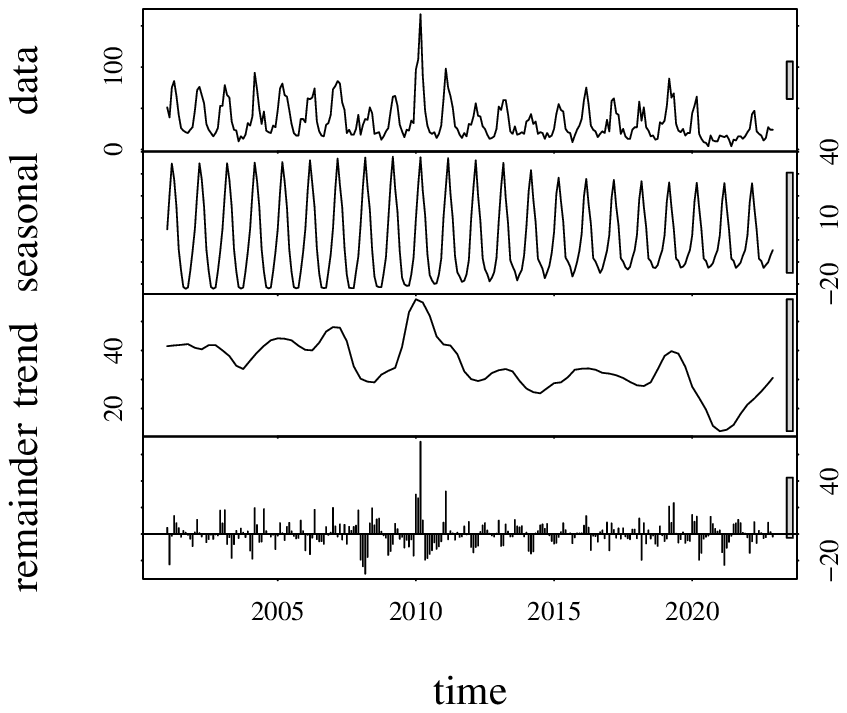}\label{fig_y2_stl_plot}
	}
	\subfigure[Monthplot of the number of hospitalizations by leptospirosis.]{\includegraphics[width=.42\linewidth]{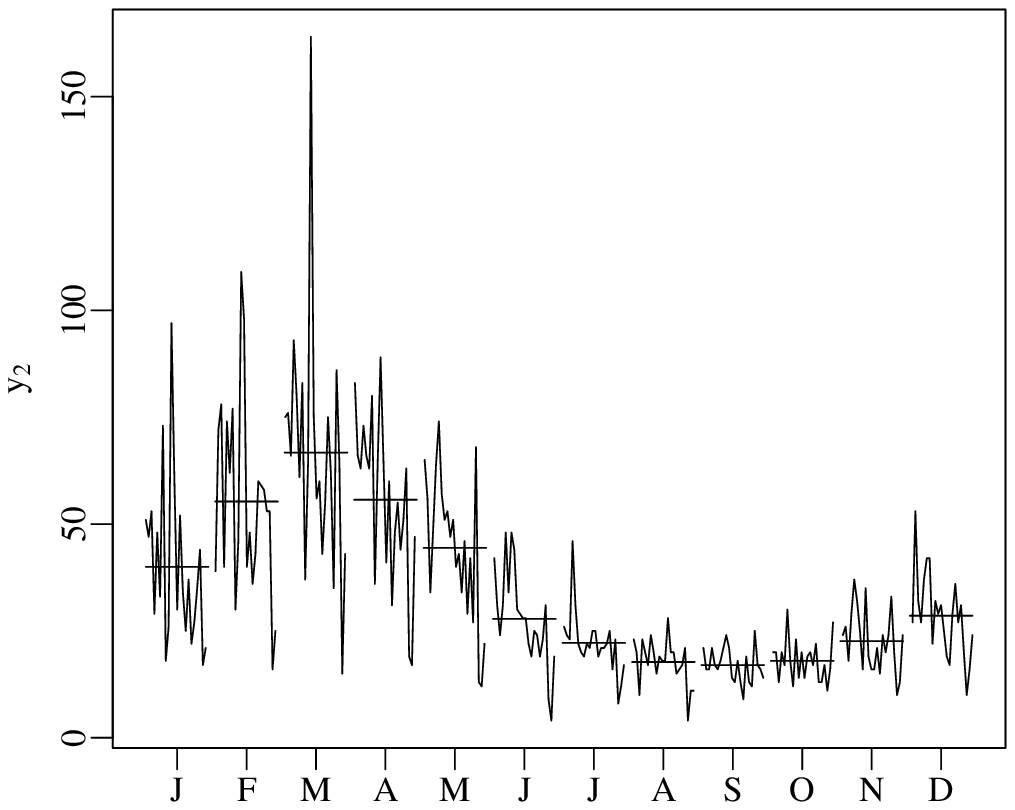}
		\label{fig_y2_monthplot}}
		\label{fig_y1_leptosp}
	\caption{Seasonal trend decomposition and monthplot of the number of observed cases and hospitalizations by leptospirosis in São Paulo, Brazil.} 
\end{figure}


Figure~\ref{fig_y1_y2_plots} combines both the series of the number of cases and hospitalizations due to leptospirosis in a single graph, enabling a visual analysis of their relationship.
Notice that an increase in the number of cases increases hospitalizations series.
Moreover, for visually analyzing the correlation between these time series, we present the CCF plot in Figure~\ref{fig_y1_y2_ccf}.
We can observe that the highest correlation is at lag $-1$, corresponding to a positive correlation. Hence, an increase in leptospirosis cases appears to lead to an increase in hospitalizations a month later.

\begin{figure}[!htpb]%
	\centering
	\subfigure[Observed cases number ($y_1$) and hospitalizations by leptospirosis ($y_2$).]{
		\includegraphics[width=.48\linewidth]{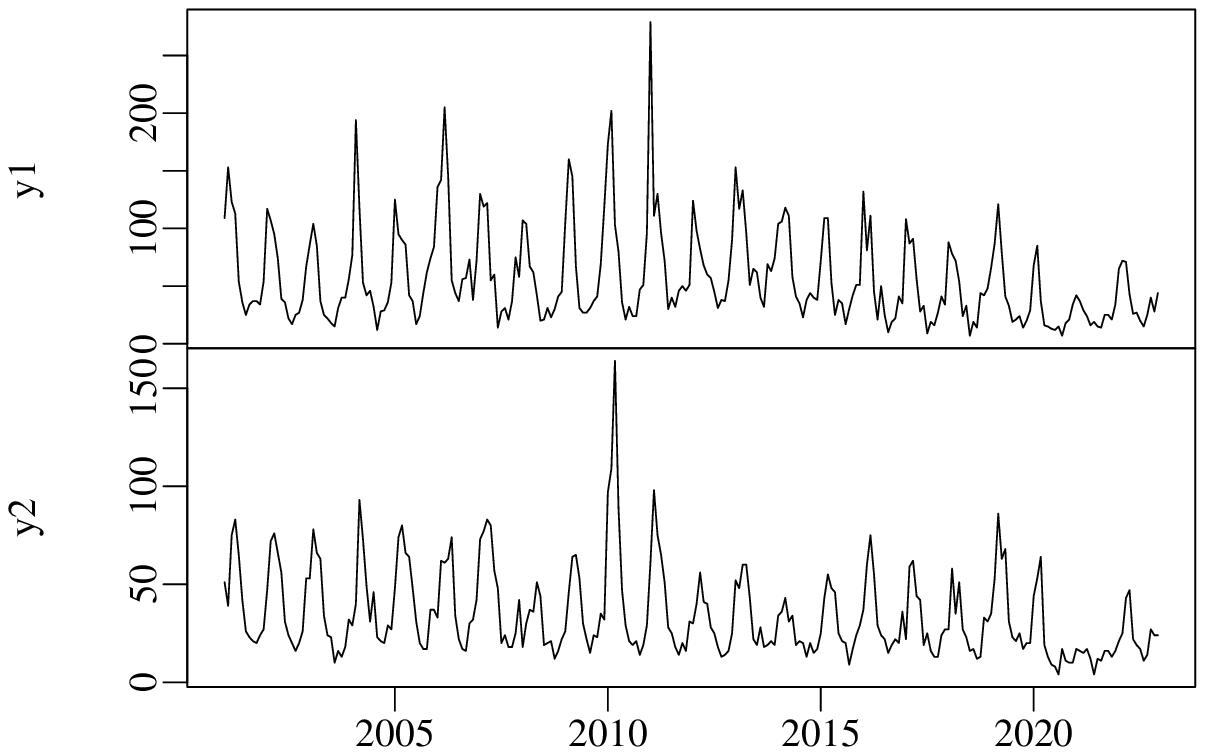}\label{fig_y1_y2_plots}
	}
	\subfigure[Cross-correlation function plot between $y_1$ and $y_2$, respectively.]{\includegraphics[width=.44\linewidth]{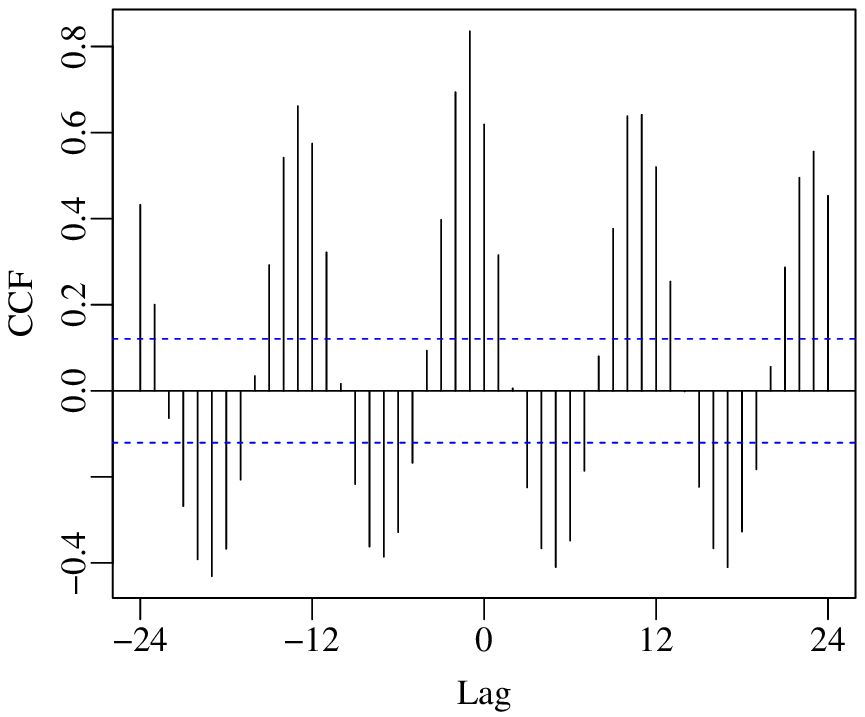}
		\label{fig_y1_y2_ccf}}
	\caption{Plots of the observed cases number and hospitalizations due to leptospirosis in São Paulo, Brazil and cross-correlation function plot.} 
\end{figure}

%
%

%
%

We selected the BGAR$(9,0,2,$
$2)$ model with some autoregressive coefficients being null and Negative Binomial (NB) marginal distributions for both time series of the number of leptospirosis cases ($\bm{y}_1$) and hospitalizations due to leptospirosis ($\bm{y}_2$), referred to as the NB-NB-BGAR$(9,0,2,2)$ model. 
We considered the AIC criterion, parameter significance, and residual analysis to select these orders. High orders (up to three) were selected based on the significance of lags in the residual ACF.
We consider the logarithm link function for $g_k(\cdot)$. The monthly seasonality effect is represented by dummy variables corresponding to each month of the year. We consider August and September as the reference categories for the dummy variables associated with $\bm{y}_1$ and $\bm{y}_2$, respectively. Additionally, we also incorporate trend effects to capture the dynamics of both time series $(\bm{y}_1,\bm{y}_2)^\top$. The autoregressive parameters $\phi_{11, l}$ for $l=3,4,6,7,8$ are not considered since they were not statistically significant at the $5\%$ significance level as well as trend effects, which presented $p$-values equal to $0.1012$ and $0.0629$, respectively. The dynamic component for $(\bm{\mu}_1,\bm{\mu}_2)^\top$ is expressed as
\begin{align}\label{eq_mod_app_lep}
	&	\log(\mu_{1t})=\eta_{1t}=
	\bm{x}_{1t}^\top\bm{\beta_{1}}+
	\sum_{l\in\{1,2,5,9\}}\phi_{11,l}
	[\log(y_{1t-l})-\bm{x}_{1t-l}^\top\bm{\beta_{1}}]
	\\ \nonumber
	&
	\log(\mu_{2t})=\eta_{2t}=
	\bm{x}_{2t}^\top\bm{\beta_{2}}+
	\sum_{l=1}^{2}\phi_{22,l}
	[\log(y_{2t-l})-\bm{x}_{2t-l}^\top\bm{\beta_{2}}]+
	\sum_{l=1}^{2}
	\phi_{21,l}
	[\log(y_{1t-l})-\bm{x}_{1t-l}^\top\bm{\beta_{1}}],
\end{align}
where 
$\bm{x}_{1t}=
(
\mathrm{jan}_t,
\mathrm{feb}_t,
\mathrm{mar}_t,
\mathrm{apr}_t,
\mathrm{may}_t,
\mathrm{jun}_t,
\mathrm{oct}_t,
\mathrm{nov}_t,
\mathrm{dec}_t
)^\top$
and
$\bm{x}_{2t}=
(
\mathrm{jan}_t,
\mathrm{feb}_t,
\mathrm{mar}_t,
\mathrm{apr}_t,
\mathrm{may}_t$,
$\mathrm{jun}_t,
\mathrm{jul}_t,$
$\mathrm{nov}_t,
\mathrm{dec}_t
)^\top$ 
are the vector of covariates associated to the leptospirosis cases ($y_{1t}$) and hospitalizations ($y_{2t}$) series, respectively.
The covariates 
$\mathrm{jan}_t$,
$\mathrm{feb}_t$,
$\mathrm{mar}_t$,
$\mathrm{abr}_t$,
$\mathrm{may}_t$,
$\mathrm{jun}_t$,
$\mathrm{jul}_t$,
$\mathrm{oct}_t$,
$\mathrm{nov}_t$,
and
$\mathrm{dec}_t$,
are indicator variables to January, February, March, April, May, June, July, October, and November, respectively. They are equal to one if the observation corresponds to these months and zero otherwise. We chose these months because the number of cases and hospitalizations are higher in these months, as verified by the analysis from Figures~\ref{fig_y1_monthplot} and~\ref{fig_y2_monthplot}. Furthermore, as expected, the remaining months were not statistically significant at the $5\%$ significance level.
 Note that the autoregressive parameters $\bm{\phi_{21}}=(\phi_{21,1},\phi_{21,2})^\top$ in~\eqref{eq_mod_app_lep} will allow measuring the shock effect of the number of leptospirosis cases on hospitalizations due to leptospirosis. 
 As expected, the reverse effect (parameters $\phi_{12,l}$, $l=1,\ldots,p_{12}$) was not significant.

In Table~\ref{tab_fitted_model_app1}, we present parameter estimates, standard errors, $z$ statistic values, and associated $p$-values for the fitted NB-NB-BGAR$(9,0,2,2)$ model. All parameter estimates are significant at a $5\%$ significance level. The estimates for the parameter $\kappa$ of the NB distribution were
$\tilde{\kappa}_1 = 14$
and
$\tilde{\kappa}_2 = 22$
for the cases and hospitalizations series, respectively. All monthly effects are positive, corresponding to the months with the highest number of leptospirosis cases and hospitalizations compared to the reference months. Specifically, January and March have the highest number of cases and hospitalizations due to leptospirosis in São Paulo, Brazil.

Notice that the estimates of the autoregressive parameters $\bm{\hat{\phi}_{21}}=(0.4059,-0.1307)^\top$ capture the lagged effect of the number of leptospirosis cases on hospitalizations. These estimates are shock effects that indicate how the lagged number of cases influences the hospitalization series. 
An increase of one case at time $t$ leads to an increase in hospitalizations in the subsequent month by a factor of $\exp(0.4059) = 1.50$, although this effect is attenuated in the second month, with a multiplicative factor of $\exp(-0.1307) = 0.87$.

\begin{table}[!htbp] 
	\centering 
	\caption{Parameter estimates, standard erros, $z$ statistic values, and associated $p$-values from NB-NB-BGAR($9,0,2,2$) model.} 
	\label{tab_fitted_model_app1} 
	\begin{tabular}{@{\extracolsep{5pt}} lcrrrr} 
		\\[-1.8ex]
		\hline  \\[-1.8ex] 
	Effect & Parameter	& Estimate & Std. Error & $z$ value & Pr(\textgreater \textbar $z$\textbar ) \\ 
	\hline \\[-1.8ex] 
    \multicolumn{6}{c}{Number of cases}\\
		\hline\\[-1.8ex] 
Intercept & $\beta_{10}$ & $3.4727$ & $0.1787$ & $19.4281$ & $<0.0001$ \\ 
January   & $\beta_{11}$ & $1.4282$ & $0.0840$ & $16.9943$ & $<0.0001$ \\ 
February  & $\beta_{12}$ & $1.4086$ & $0.0820$ & $17.1871$ & $<0.0001$ \\ 
March     & $\beta_{13}$ & $1.3271$ & $0.0842$ & $15.7573$ & $<0.0001$ \\ 
April     & $\beta_{14}$ & $0.8886$ & $0.0826$ & $10.7546$ & $<0.0001$ \\ 
May       & $\beta_{15}$ & $0.3663$ & $0.0818$ & $4.4789$ & $<0.0001$ \\ 
June      & $\beta_{16}$ & $0.2546$ & $0.0754$ & $3.3747$ & $0.0007$ \\ 
October   & $\beta_{17}$ & $0.3509$ & $0.0732$ & $4.7975$ & $<0.0001$ \\ 
November  & $\beta_{18}$ & $0.5058$ & $0.0789$ & $6.4084$ & $<0.0001$ \\ 
December  & $\beta_{19}$ & $0.7561$ & $0.0819$ & $9.2301$ & $<0.0001$ \\ 
$-$       & $\phi_{11,1}$ & $0.3825$ & $0.0629$ & $6.0783$ & $<0.0001$ \\ 
$-$       & $\phi_{11,2}$ & $0.1789$ & $0.0636$ & $2.8122$ & $0.0049$ \\ 
$-$       & $\phi_{11,5}$ & $0.1500$ & $0.0617$ & $2.4312$ & $0.0151$ \\ 
$-$       & $\phi_{11,9}$ & $0.1475$ & $0.0564$ & $2.6145$ & $0.0089$ \\ 
\hline \\[-1.8ex] 
\multicolumn{6}{c}{Hospitalizations}\\
\hline\\[-1.8ex] 
Intercept & $\beta_{20}$ & $3.0793$ & $0.1479$ & $20.8255$ & $<0.0001$ \\ 
January   & $\beta_{21}$ & $0.7814$ & $0.0880$ & $8.8764$ & $<0.0001$ \\ 
February  & $\beta_{22}$ & $1.1435$ & $0.0898$ & $12.7277$ & $<0.0001$ \\ 
March     & $\beta_{23}$ & $1.2884$ & $0.0903$ & $14.2662$ & $<0.0001$ \\ 
April     & $\beta_{24}$ & $1.0898$ & $0.0909$ & $11.9838$ & $<0.0001$ \\ 
May       & $\beta_{25}$ & $0.8795$ & $0.0896$ & $9.8187$ & $<0.0001$ \\ 
June      & $\beta_{26}$ & $0.3995$ & $0.0858$ & $4.6549$ & $<0.0001$ \\ 
July      & $\beta_{27}$ & $0.2220$ & $0.0797$ & $2.7859$ & $0.0053$ \\ 
November  & $\beta_{28}$ & $0.2505$ & $0.0758$ & $3.3048$ & $0.0010$ \\ 
December  & $\beta_{29}$ & $0.4784$ & $0.0811$ & $5.8985$ & $<0.0001$ \\ 
$-$       & $\phi_{22,1}$ & $0.3924$ & $0.0700$ & $5.6074$ & $<0.0001$ \\ 
$-$       & $\phi_{22,2}$ & $0.2280$ & $0.0635$ & $3.5927$ & $0.0003$ \\ 
$-$       & $\phi_{21,1}$ & $0.4059$ & $0.0559$ & $7.2641$ & $<0.0001$ \\ 
$-$       & $\phi_{21,2}$ & $-0.1307$ & $0.0618$ & $-2.1135$ & $0.0346$ \\ 
\hline
& & & & AIC& $= 3959.83$  \\
& & & & BIC& $=4058.65$\\
		\hline \\[-1.8ex] 
	\end{tabular} 
\end{table}

Figure~\ref{fig_y1_y2_residual_plots} displays residual plots from the fitted NB-NB-BGAR($9,0,2,2)$ model. 
Figures~\ref{fig_res_y1_QQplots} and~\ref{fig_res_y2_QQplots} display the QQ-plots from residuals of the fitted model for the leptospirosis cases and hospitalizations series, respectively. These plots suggest that the assumption of residual normality is satisfied. The plots of the residuals versus (vs.) indexes in Figures~\ref{fig_res_ind_y1} and~\ref{fig_res_ind_y2} show that there are not any systematic trend or pattern in the residuals, indicating that the model captures all the underlying structure in the data. 
The residual ACFs are shown in Figures~\ref{fig_Acf_res_y1} and~\ref{fig_Acf_res_y2}. These correlograms provide evidence that the residuals are not autocorrelated. To confirm this residual analysis, we perform some hypothesis tests. We consider the Shapiro-Wilk and Ljung–Box test to check the assumption of normality and
serial correlation of the residuals. We consider $20$ lags for the Ljung–Box test and degrees of freedom $16$ and $18$ for the leptospirosis cases and hospitalizations series, respectively.
We set a $5\%$ significance level. Table~\ref{tab_sw_lb_tests} presents the $p$-value obtained from these tests for each case. Since all $p$-values are greater than $0.05$, we do not reject the residuals' normality and uncorrelated null hypothesis. According to the results, there is evidence that the model fits these data well.

\begin{table}[!htp] 
	\centering 
	\caption{$p$-values from Shapiro-Wilk and Ljung–Box tests for the residuals of the fitted model for the leptospirosis cases and hospitalizations. } 
	\label{tab_sw_lb_tests} 
	\begin{tabular}{@{\extracolsep{5pt}} lrr} 
		\hline 
		\\[-1.8ex] 
		$p$-value & Cases & Hospitalizations \\
		\hline 	\\[-1.8ex] 
		Shapiro-Wilk &  $0.4870$ & $0.7706$ \\
		Ljung–Box & $0.0750$  & $0.2605$\\
		\hline 
	\end{tabular} 
\end{table}

\begin{figure}[!htpb]%
	\centering
	\subfigure[QQ-plots of the residuals from the fit of the hospitalizations due to leptospirosis time series.]{
		\includegraphics[width=.42\linewidth]{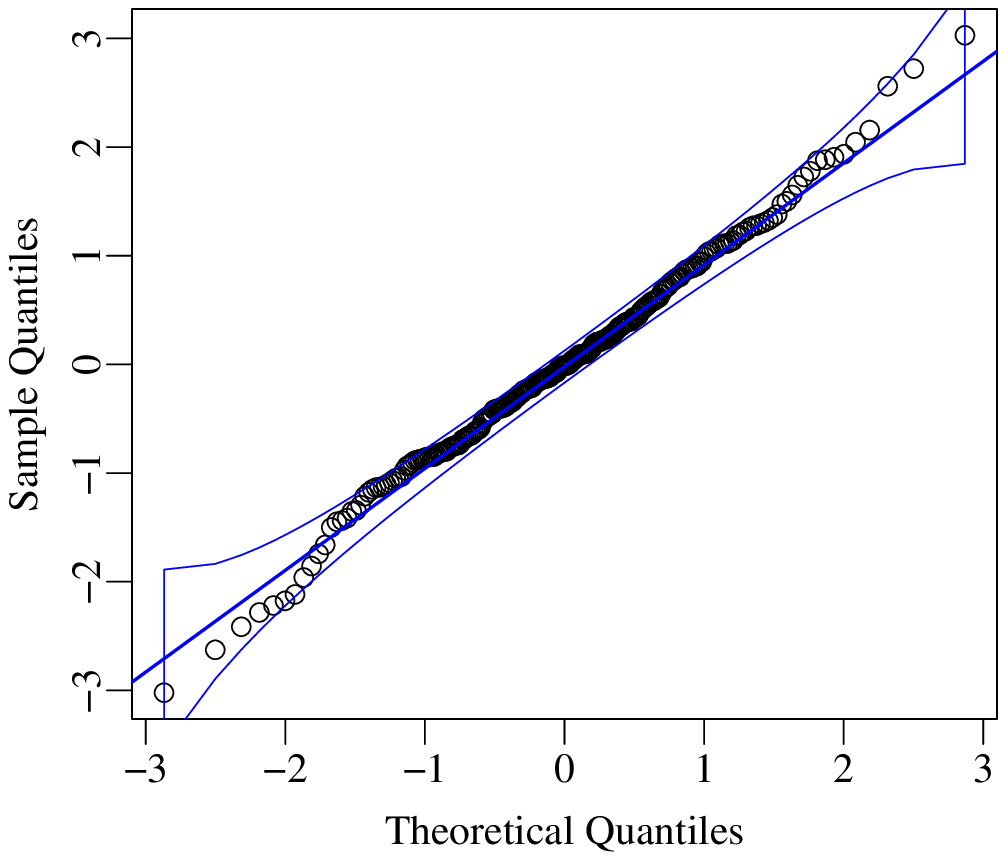}
		\label{fig_res_y1_QQplots}
	}
	\subfigure[QQ-plots of the residuals from the fit of the leptospirosis number of cases time series.]{
	\includegraphics[width=.42\linewidth]{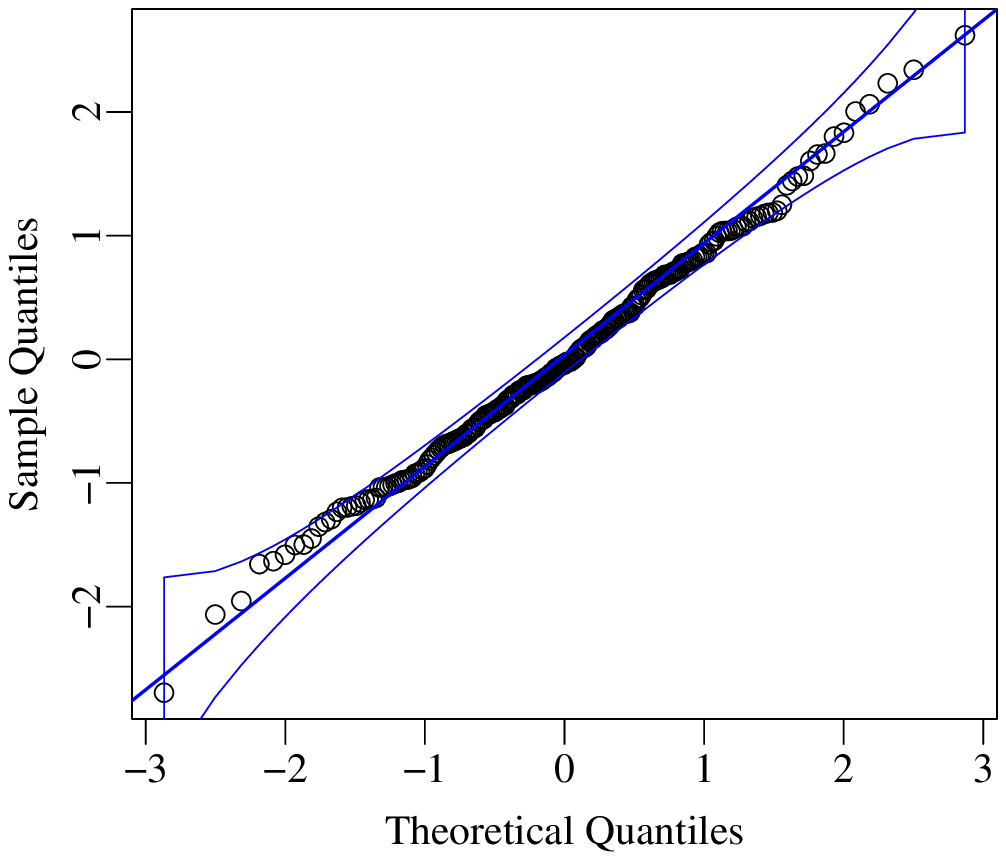}
	\label{fig_res_y2_QQplots}
}
	\subfigure[Residuals vs. indexes from the fit of the hospitalizations due to leptospirosis time series.]{
	\includegraphics[width=.42\linewidth]{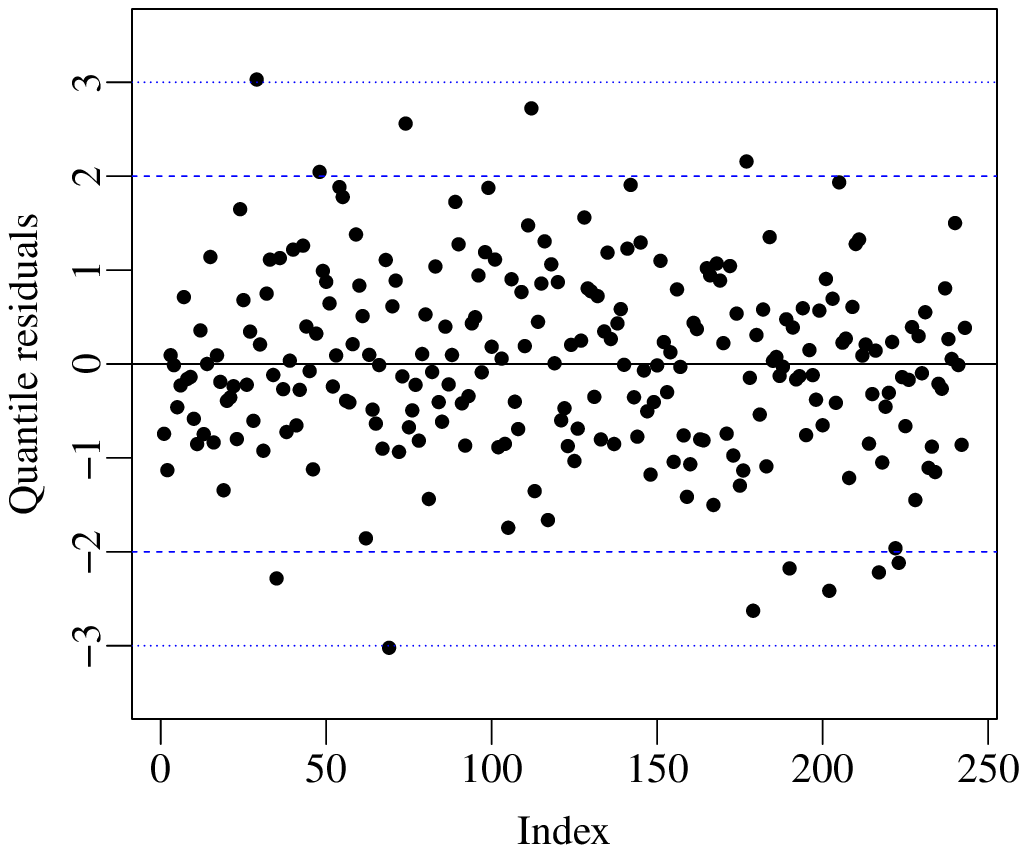}
	\label{fig_res_ind_y1}
}
\subfigure[Residuals vs. indexes from the fit of the leptospirosis number of cases time series.]{
	\includegraphics[width=.42\linewidth]{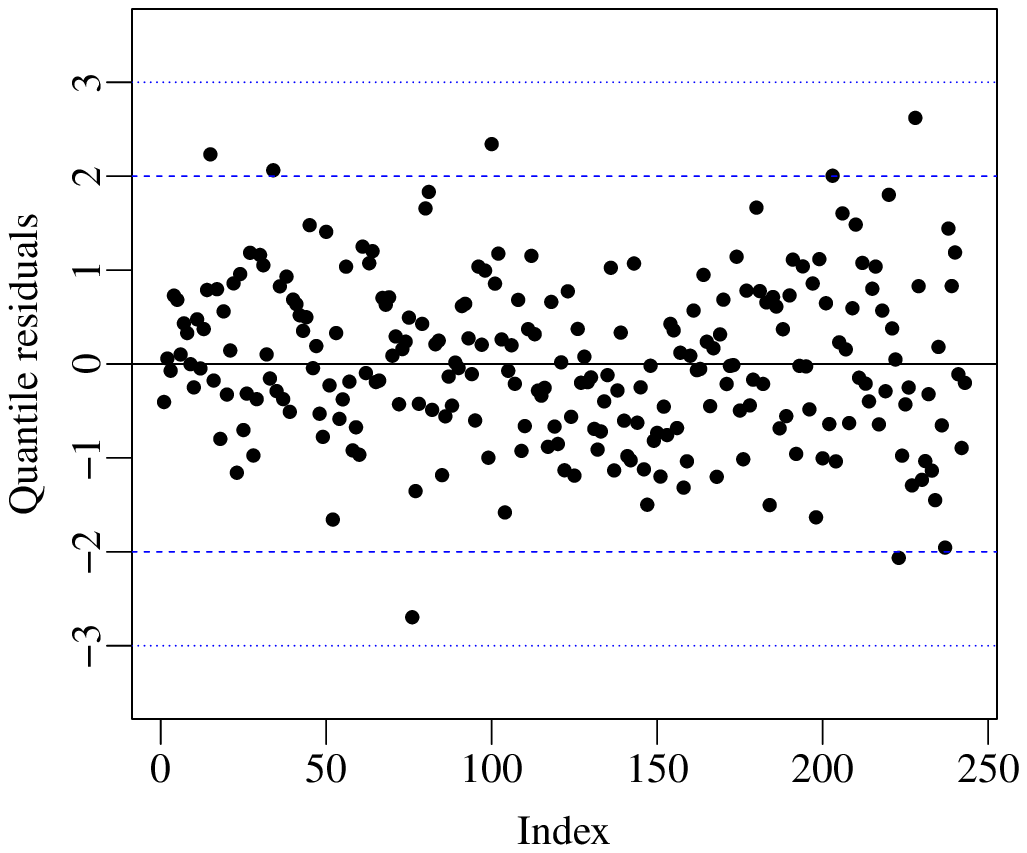}
	\label{fig_res_ind_y2}
}
	\subfigure[ACF of the residuals from the fit of the hospitalizations due to leptospirosis time series.]{
	\includegraphics[width=.42\linewidth]{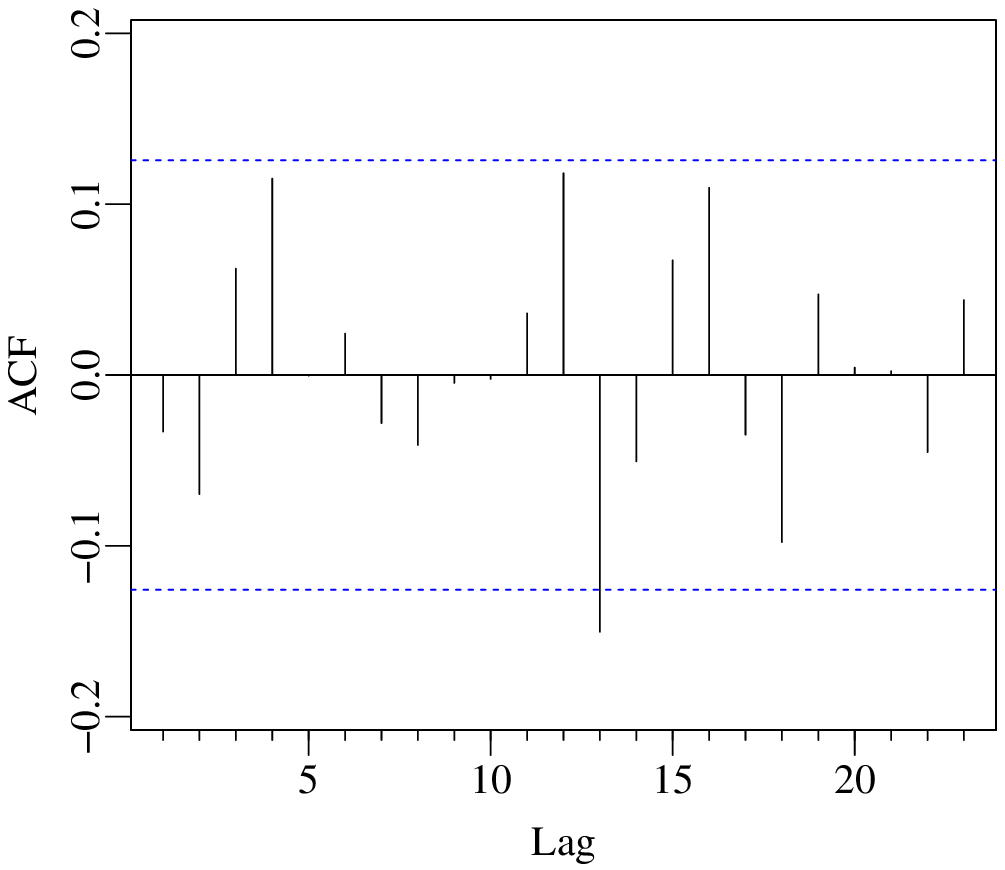}
	\label{fig_Acf_res_y1}
}
	\subfigure[ACF of the residuals from the fit of the leptospirosis number of cases time series.]{
	\includegraphics[width=.42\linewidth]{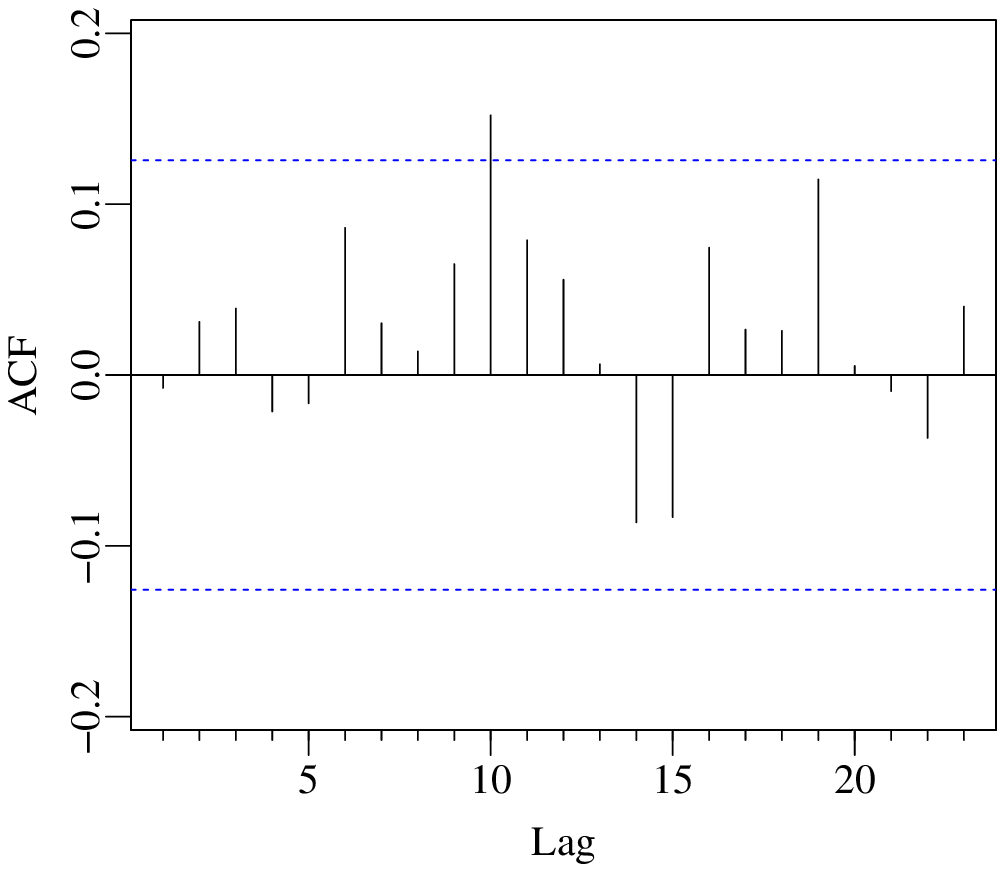}
	\label{fig_Acf_res_y2}
}
	\caption{Residual plots from fits of number of cases and hospitalizations times series.} 
	\label{fig_y1_y2_residual_plots}
\end{figure}

 Figure~\ref{fig_comp_res_index_y1y2} shows the composite residuals versus the index plot. As expected, most points are below the horizontal line (the 95th quantile of the chi-squared distribution with two degrees of freedom). This indicates that the composite residuals corroborate with the model assumptions, as they follow the expected chi-squared distribution without systematic deviations or patterns over time.
 Figure~\ref{fig_ccf_res1res2} displays the CCF plot of residuals from the NB-NB-BGAR($9,0,2,2$) model. The absence of significant correlations in the CCF plot of residuals suggests that the NB-NB-BGAR$(9,0,2,2)$ model has effectively captured the time dependence between the series.

\begin{figure}[!htpb]%
	\centering

	\subfigure[Composite residuals vs. indexes from NB-NB-BGAR($9,0,2,2)$ model.]{
	\includegraphics[width=.42\linewidth]{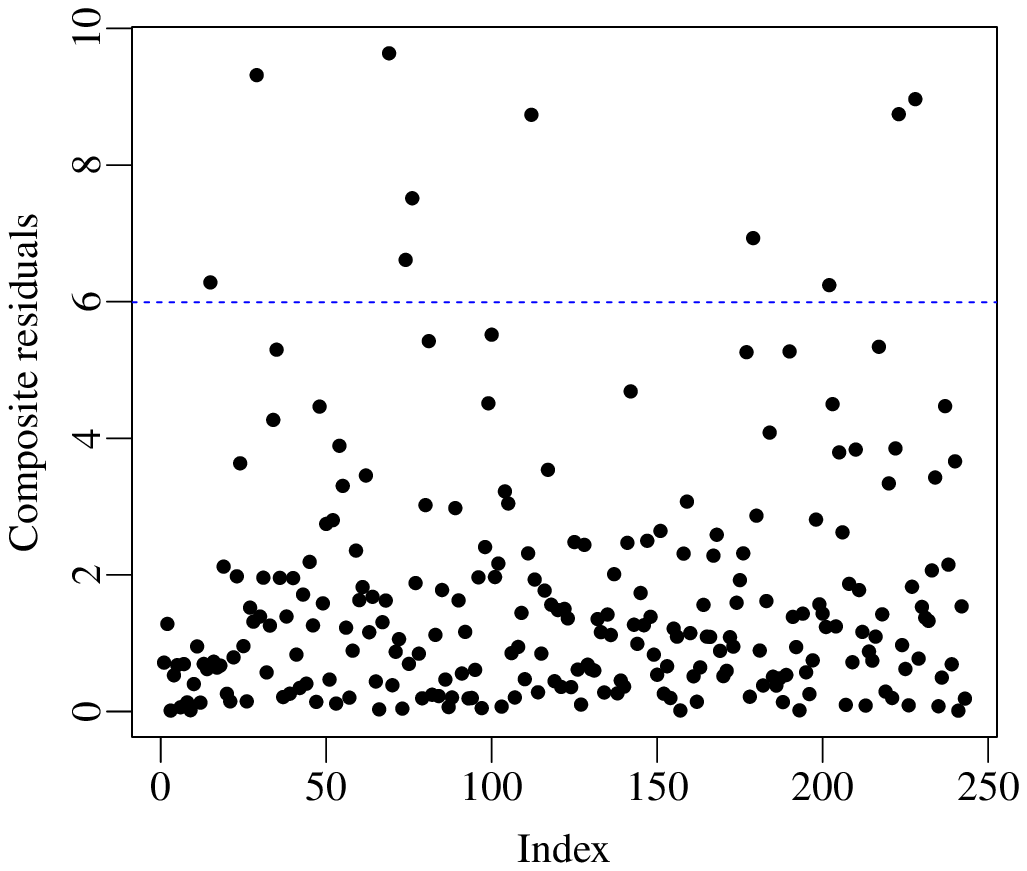}
	\label{fig_comp_res_index_y1y2}
}
	\subfigure[CCF plot of residuals from NB-NB-BGAR($9,0,2,2)$ model.]{
	\includegraphics[width=.42\linewidth]{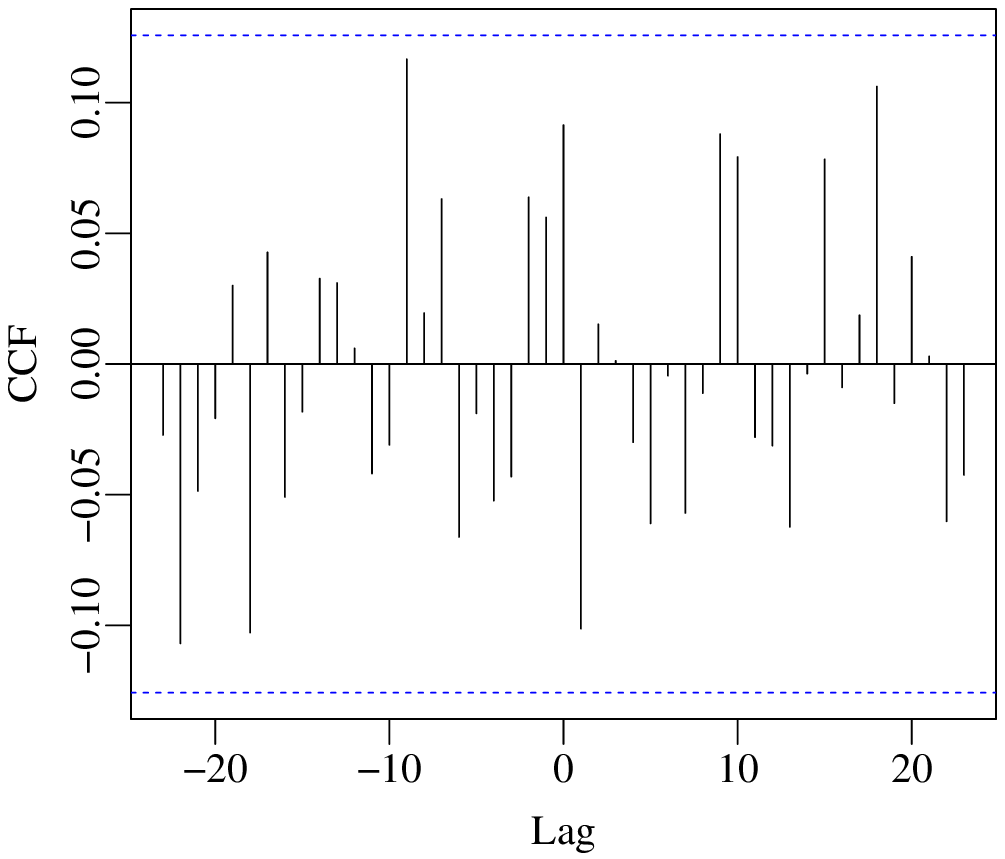}
	\label{fig_ccf_res1res2}
}
	\caption{Composite residuals vs. indexes and CCF plot of residuals from NB-NB-BGAR($9,0,2,2)$ model.} 
	\label{fig_comp_res_ccf_res1res2}
\end{figure}

Figure~\ref{fig_y1_y2_predicted} compares predicted and observed values for each time series: number of cases and hospitalizations due to leptospirosis. The predicted values are very close to the observed values for both series. According to this figure, the proposed model provides a good in-sample fit for the dataset.
Moreover, it displays the forecasted values for the last $12$ months, which we separated from the original sample.
We compute the forecasts as described in the schemes in Section~\ref{sec_predict_forec}.
Note that the out-of-sample forecasts are also quite close to the observed values, indicating that the NB-NB-BGAR($9,0,2,2$) model provided accurate forecasts.


\begin{figure}[!htpb]
	\centering
	\includegraphics[width=1.02\textwidth]{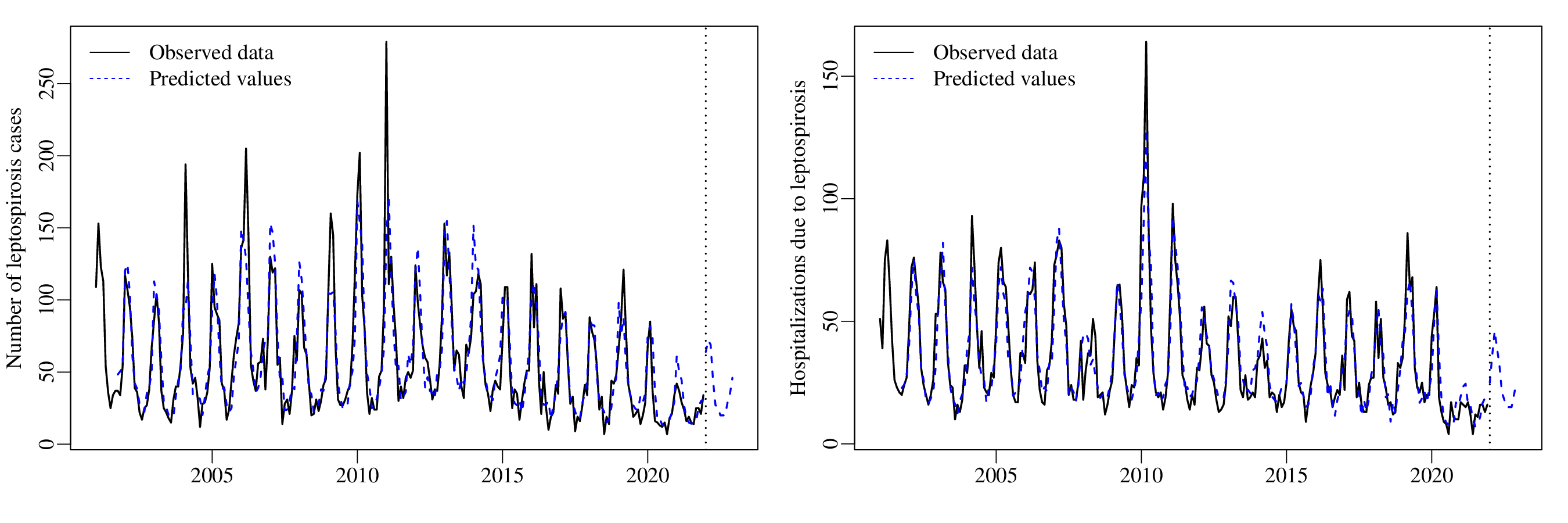}
	\caption{Observed number of leptospirosis cases and hospitalizations, in-sample and out-of-sample predicted values from the NB-NB-BGAR($9,0,2,2$) model.}
	\label{fig_y1_y2_predicted}
\end{figure}

For comparison purposes, we fit the Negative Binomial GARMA model for each one of the time series, as well as the classical ARIMA model. We also consider the VAR model for this dataset. Within the GARMA class, the selected models were NB-GARMA$(9,0)$ (with $\phi_i=0$ for $i=3,4,6,7,8$) and NB-GARMA$(1,1)$ models for the number of leptospirosis cases and hospitalizations, respectively, according to AIC criterion and residual analysis. We use the \verb|auto.arima| function of the \verb|forecast| package~\cite{hyndman2015forecasting,hyndman2008automatic}, based on the AIC, to select the best model in the ARIMA class. The ARIMA model did not fit the data well on its original scale. Then, we applied the logarithm transformation to the response variables to achieve a better fit. The selected models were SARIMA$(1,1,1)(0,0,1)_{12}$ and ARIMA$(0,1,4)$ for the number of leptospirosis cases and hospitalizations series, respectively. The covariates defined by $\bm{x}_{1t}$ and $\bm{x}_{2t}$ in~\eqref{eq_mod_app_lep}
are incorporated in the fits of the GARMA and ARIMA models. In the VAR model, we included covariates corresponding to all months of the year, including August, as the reference category. 
Only the NB-GARMA$(9,0)$ and NB-GARMA$(2,0)$ models successfully pass all residual diagnostic checks. Nevertheless, we conduct out-of-sample forecasting using the ARIMA and VAR models just for comparative analysis.

Table~\ref{tab_forec_error_app_lep} displays each fitted model's RMSE, MAE, and MAPE values for the considered forecast horizons $h=1,\ldots,12$. 
The NB-GARMA$(9,0)$ provides the best forecasts for the number of leptospirosis cases series since it presents the smallest RMSE, MAE, and MAPE values at $h=1,\ldots,12$. On the other hand, for the hospitalizations series, the NB-NB-BGAR($9,0,2,2$) provides the best forecasts in the most considered forecast horizons.
These results demonstrate that the proposed model outperforms alternative models in out-of-sample forecasting for leptospirosis-related hospitalizations in São Paulo state, Brazil. This application highlights the practical relevance of the BGAR model, as accurate leptospirosis forecasting is essential for optimizing hospital capacity and enhancing public health responses.

\begin{table}[!htbp] 
	\renewcommand{\arraystretch}{1.3}
	\centering 
	\scriptsize
	\caption{ Comparison of out-of-sample forecasting errors among fitted models within each class.} 
	\label{tab_forec_error_app_lep} 
		\rotatebox{90}{
	\begin{tabular}{@{\extracolsep{5pt}} l|lrrrrrrrrrrrr} 
		\hline 
		Model &Measures & $h=1$ &  $h=2$  & $h=3$  & $h=4$  & $h=5$  & $h=6$  & $h=7$  & $h=8$ & $h=9$ & $h=10$ & $h=11$& $h=12$ \\
			\hline
		\multicolumn{13}{c}{Number of cases}\\
		\hline
   \multirow{3}{*}{NB-NB-BGAR($9,0,2,2$)}   & RMSE  & $0.5688$ & $0.4204$ & $1.2968$ & $1.2226$ & $1.1491$ & $1.0997$ & $1.0322$ & $1.9986$ & $2.4202$ & $4.1262$ & $4.5220$ & $4.3874$ \\ 
                                            & MAE  & $0.5688$ & $0.3708$ & $0.9692$ & $0.9685$ & $0.9327$ & $0.9120$ & $0.8460$ & $1.3589$ & $1.7142$ & $2.6269$ & $3.0603$ & $3.0104$ \\ 
                                            & MAPE  & $0.8752$ & $0.5575$ & $1.3886$ & $1.6033$ & $1.8900$ & $2.0740$ & $2.0990$ & $5.9612$ & $7.3239$ & $9.3019$ & $10.8570$ & $10.4185$ \\ 
\hline
\multirow{3}{*}{NB-GARMA($9,0$)}  & RMSE & $7.1661$ & $7.7481$ & $9.2535$ & $8.5068$ & $7.7554$ & $7.3472$ & $6.8998$ & $6.4676$ & $6.7217$ & $8.3484$ & $7.9603$ & $7.9752$ \\ 
                                  & MAE  & $7.1661$ & $7.7277$ & $9.0507$ & $8.2150$ & $7.2434$ & $6.8382$ & $6.2984$ & $5.6583$ & $5.9724$ & $7.0790$ & $6.4595$ & $6.5993$ \\ 
                                  & MAPE & $11.0248$ & $11.2689$ & $13.0040$ & $13.0715$ & $13.0396$ & $13.8369$ & $14.0453$ & $13.2711$ & $15.5676$ & $18.2707$ & $16.6955$ & $16.8452$ \\ 
		\hline
			\multirow{3}{*}{SARIMA$(1,1,1)\times(0,0,1)_{12}$} & RMSE & $5.4578$ & $10.0149$ & $13.2443$ & $12.0118$ & $10.9644$ & $10.5386$ & $9.9588$ & $9.3168$ & $9.3145$ & $10.5657$ & $10.1236$ & $10.2198$ \\ 
	& 	MAE & $5.4578$ & $9.2636$ & $12.1909$ & $10.9267$ & $9.7204$ & $9.4469$ & $8.8515$ & $7.7993$ & $7.9656$ & $9.0008$ & $8.4841$ & $8.7124$ \\ 
		& MAPE & $8.3966$ & $13.2742$ & $17.3215$ & $17.1390$ & $17.4767$ & $19.5512$ & $20.5289$ & $18.3247$ & $20.4201$ & $22.9573$ & $21.9475$ & $22.2443$ \\ 
		\hline
			\multirow{3}{*}{VAR($2$)} &RMSE & $34.6509$ & $31.7160$ & $29.6794$ & $27.8860$ & $25.3689$ & $23.2418$ & $21.6147$ & $20.4522$ & $19.3247$ & $18.3389$ & $18.1767$ & $17.8737$ \\ 
			& MAE & $34.6509$ & $31.5656$ & $29.4156$ & $27.4695$ & $24.0483$ & $20.8427$ & $18.6379$ & $17.3977$ & $15.8899$ & $14.4486$ & $14.6321$ & $14.5892$ \\ 
			& MAPE & $53.3091$ & $46.4325$ & $42.7463$ & $44.6361$ & $43.6806$ & $39.3726$ & $37.6117$ & $40.1740$ & $37.4113$ & $34.0394$ & $36.2912$ & $35.9408$ \\ 
		\hline
		\multicolumn{13}{c}{Hospitalizations}\\
		\hline
        \multirow{3}{*}{NB-NB-BGAR($9,0,2,2$)} & RMSE & $3.9050$ & $8.2592$ & $6.9168$ & $7.2736$ & $7.9534$ & $7.2883$ & $6.7606$ & $6.4625$ & $6.1000$ & $6.8888$ & $6.6971$ & $6.4225$ \\ 
                                               & MAE  & $3.9050$ & $7.4566$ & $5.8590$ & $6.4572$ & $7.2119$ & $6.2698$ & $5.5322$ & $5.3113$ & $4.8187$ & $5.5186$ & $5.4112$ & $5.0664$ \\ 
                                               & MAPE & $18.5950$ & $31.3139$ & $22.9409$ & $21.5951$ & $26.5765$ & $23.5147$ & $21.0854$ & $22.7283$ & $20.8994$ & $23.1865$ & $22.7215$ & $21.2703$ \\ 
\hline
			\multirow{3}{*}{NB-GARMA($1,1$)}& RMSE  & $6.5080$ & $11.9483$ & $10.8307$ & $9.4498$ & $11.0354$ & $10.2903$ & $9.6063$ & $9.1921$ & $8.7184$ & $8.8345$ & $8.4403$ & $8.1900$ \\ 
                                            & MAE   & $6.5080$ & $11.0509$ & $10.0832$ & $8.1373$ & $9.6828$ & $8.9261$ & $8.1167$ & $7.7867$ & $7.2383$ & $7.4963$ & $6.9756$ & $6.7789$ \\ 
                                            & MAPE  & $30.9902$ & $46.6830$ & $37.4381$ & $29.3017$ & $37.8641$ & $36.0644$ & $33.6523$ & $35.6686$ & $33.9682$ & $34.2078$ & $31.7680$ & $30.7233$ \\ 
        \hline
		\multirow{3}{*}{ARIMA($0,1,4$)} & RMSE & $1.8186$ & $1.6926$ & $6.3553$ & $11.3956$ & $10.1932$ & $9.5938$ & $9.1435$ & $8.5829$ & $8.2636$ & $9.6928$ & $9.9963$ & $9.9650$ \\ 
		&MAE & $1.8186$ & $1.6875$ & $4.7064$ & $8.5190$ & $6.8681$ & $6.6770$ & $6.5436$ & $5.9789$ & $5.8729$ & $7.0882$ & $7.5925$ & $7.7610$ \\ 
		&MAPE & $8.6600$ & $7.4429$ & $13.2908$ & $20.5834$ & $16.7073$ & $18.9417$ & $21.0620$ & $20.7308$ & $22.4157$ & $26.8502$ & $29.1958$ & $30.1010$ \\
	    \hline
	 	\multirow{3}{*}{VAR($2$)} & RMSE & $1.7579$ & $2.1893$ & $10.6603$ & $13.9603$ & $12.6347$ & $11.8721$ & $11.4150$ & $11.6798$ & $11.4850$ & $10.9544$ & $10.4470$ & $10.0046$ \\ 
		&MAE & $1.7579$ & $2.1533$ & $7.5031$ & $10.8632$ & $9.5535$ & $9.1100$ & $8.9730$ & $9.5248$ & $9.5540$ & $8.9569$ & $8.2100$ & $7.5886$ \\ 
		&MAPE & $8.3711$ & $9.2828$ & $20.2991$ & $26.3646$ & $25.0140$ & $26.8911$ & $29.8988$ & $41.3749$ & $44.5453$ & $41.4178$ & $37.9333$ & $35.0338$ \\ 
		\hline 
	\end{tabular} 
}
\end{table}

\section{Conclusion}\label{sec_conclusion}


This study introduced the bivariate generalized autoregressive (BGAR) model as a flexible and simplified alternative to the conventional bivariate vector autoregressive (VAR) model. 
It generalizes the bivariate VAR model, allowing the analysis of bivariate time series of count, binary, Gaussian, and non-Gaussian data and incorporating the effect of one lagged series in the other.
We used the conditional maximum likelihood method for parameter estimation. We derived general closed-form expressions for the conditional score vector and conditional Fisher information matrix. Numerical results from Monte Carlo studies were provided to assess interval and point estimation and showed the estimator's consistency on finite samples. We presented hypothesis testing, interval estimation, diagnostic techniques, model selection tools, and residual analysis. An empirical application to the two count time series is performed to show the proposed model's usefulness and applicability. The results of applying the BGAR model to a count series of leptospirosis cases and related hospitalizations revealed valuable insights into the disease dynamics and its public health impact. Our analysis showcased BGAR's ability to capture relationships between variables over time and provide accurate forecasts to guide public health policies and preventive interventions.
Furthermore, the introduced model's flexibility in accommodating various data types makes it versatile for time series analysis in diverse fields. Future research can explore extensions of BGAR to handle additional complexities and refine its applications in different domains. Overall, the BGAR model represents a valuable addition to the arsenal of time series models, offering simplicity, interpretability, and accuracy for analyzing and forecasting bivariate time series data.




\section*{Acknowledgements}
The authors are grateful for the partial financial support provided by CAPES, CNPq (grant 308578/2023-6), and FAPERGS (grant 23/02538-0),
Brazil.

\section*{Funding Statement}

This work was partially supported by the Brazilian agencies Coordenação de Aperfeiçoamento de Pessoal de Nível Superior (CAPES), Conselho Nacional de Desenvolvimento Científico e Tecnológico (CNPq) (grant 308578/2023-6), and Fundação de Amparo à Pesquisa do Estado de São Paulo (FAPESP) (grant 23/02538-0).
%
\section*{Competing interests} 
The authors have no competing interests to declare.

\bibliographystyle{johd}
\bibliography{bib_var_mod}

\newpage

\appendix

\section*{Appendices}

\section{Exponential family distributions in the canonical form}\label{sec_append_exp_fam_elements}

The probability density function of exponential family distributions in the canonical form can be expressed as 
\begin{align}\label{eq_fam_exp_canonical}
	f(y; \vartheta, \varphi) = \exp\left\{\frac{y \vartheta  - b_k(\vartheta)}{\varphi} + c(y,\varphi)\right\},
\end{align}
where $\vartheta$ and $\varphi$ are the canonical and dispersion parameters, respectively.
Many commonly used discrete and continuous probability distributions can be represented in the canonical exponential family, such as the Binomial, Negative Binomial, Poisson, Gamma, Normal, and Inverse Normal distributions. Table~\ref{tab_exp_fam_components} displays the main components of these distributions when represented in the exponential family.

\begin{center}
	\begin{table}[!htpb]
		\caption{Distributions in the canonical exponential family} 
		\label{tab_exp_fam_components}	
		\centering
		\begin{tabular}{ lllll }
			\hline
			Distribuição        & $\vartheta$               & $\varphi$   & $b(\vartheta)$                              & $V(\mu)$\\ 	\hline
			Binomial             &   $\log(\mu/(1-\mu))$            & $1$         &    $\log(1+\mathrm{e}^\vartheta)$                    & $\mu(1-\mu)$ \\
			Negative Binomial   & $\log (\mu/(\mu+\kappa))$ & $1$         & $\kappa\log(\kappa/(1-\mathrm{e}^\vartheta))$  & $\mu+\mu^2/\kappa$ \\
				Poisson             & $\log \mu$                & $1$         & $\exp\{\vartheta\}$                         & $\mu$\\
			Gamma               & $-1/\mu$                  & $\nu^{-1}$  & $-\log(-\vartheta)$                         & $\mu^2$\\
			Normal              & $\mu$                     & $\sigma^2$  & $\vartheta^2/2$                             & $1$\\
			Inverse Normal      & $-1/(2\mu^2)$             & $\phi^{-1}$ & $-\sqrt{-2\vartheta}$                       & $\mu^3$ \\
			\hline
		\end{tabular}
	\end{table}
\end{center}

In Tables~\ref{tab_clinha_cll} and~\ref{tab_cll_exp}, we present the non-null derivatives of the first and second-order of the $c(\cdot,\cdot)$ function of each distribution and the expected value of the second-order derivative of $c(\cdot,\cdot)$. All these quantities presented in Tables~\ref{tab_exp_fam_components},~\ref{tab_clinha_cll}, and~\ref{tab_cll_exp} are useful to obtain the elements of the conditional score vector and conditional Fisher information matrix defined in Section~\ref{sec_score_vector} and~\ref{sec_fisher}, respectively.

\begin{center}
	\begin{table}[!htpb]
		\caption{Derivatives of function $c(\cdot,\cdot)$} 
		\label{tab_clinha_cll}	
		\centering
		\small
		\begin{tabular}{ l|l|l }
			\hline
			Distribuição &$c^\prime(y,\varphi)$ & $c^{\prime\prime}(y,\varphi)$ \\
			\hline
			Gamma & $\varphi^{-2}[\log\varphi+\psi(\varphi^{-1})-\log y-1]$ &
			$2\varphi^{-3}[\log y - \log \varphi - \psi(\varphi^{-1})+1]+\varphi^{-3}-\varphi^{-4}\psi^{(1)}(\varphi^{-1})$\\
			Normal &  $y^2/(2\varphi^2)-1/(2\varphi)$  &  $-y^2/\varphi^3-1/(\varphi ^2)$\\
			Inverse Normal &   $1/(2\,y\,\varphi^2)-1/(2\varphi)$  &  $1/(2\varphi^2)-1/(y\,\varphi^3)$ \\
			\hline
		\end{tabular}
	\end{table}
\end{center}

\begin{center}
	\begin{table}[!htpb]
		\caption{Expected value of $c^{\prime \prime}(\cdot,\cdot)$		} 
		\label{tab_cll_exp}	
		\centering
		\begin{tabular}{ l|l }
			\hline
			Distribuição &$\mathrm{I\!E}(c^{\prime\prime}(y,\varphi))$ \\
			\hline
			Gamma &$2\varphi^{-3}[1  +\log\mu  ]+\varphi^{-3}-\varphi^{-4}\psi^{(1)}(\varphi^{-1})$\\
			Normal &  $1/(2\varphi^2)-1/\varphi^2-\mu^2/\varphi^3$  \\
			Inverse Normal &   $1/(2\varphi^2)-1/(\mu\varphi ^3)-1/\varphi^2$ \\
			\hline
		\end{tabular}
	\end{table}
\end{center}

In the following sections, we present the Binomial, Negative Binomial, Poisson, Gamma, Normal, and Inverse Normal distributions in their canonical exponential family and compute all the components summarized in Tables~\ref{tab_exp_fam_components},~\ref{tab_clinha_cll}, and~\ref{tab_cll_exp} for each distribution. 

%
\subsection{Binomial}

The binomial distribution is a discrete probability distribution that describes the number of successes in a fixed number of independent trials, each with the same success probability. 
Here, we consider the parameterization in terms of mean proportion of successes. For this, we use the following transformation $y^\star=y/n$, where $Y$ follows the binomial distribution of parameter $\mu$ and $n$ is the number of trials. The probability function of $y^\star$ is given by
\begin{align}\label{eq_pdf_bin}
	f(y^\star;\mu) = \binom{n}{ny^\star}  \mu ^{ny^\star} \left( 1 -\mu\right)^{n - ny^\star},
\end{align}
where
$ n $ is the number of trials,
 $ 0<\mu<1$ is the mean proportion of successes,
 $ 0<y^\star<1$ is the proportion of successes, and $ \binom{n}{ny^\star} = \frac{n!}{y^\star!(n - ny^\star)!} $ is the binomial coefficient.

For representing~\eqref{eq_pdf_bin} in the canonical exponential family,
we must assume that $n$ is known and the dispersion parameter is fixed at one; see~\cite{nelder_Wedderburn_1972_glm}. Hence, we can rewrite \eqref{eq_pdf_bin} as
\begin{align}\label{eq_pdf_bin}
	f(y^\star;\mu) = 
	\mathrm{exp}
	\left\{
	n\left[y^\star\log\left(\frac{\mu}{1-\mu}\right)+ \log	\left( 1 -\mu\right)\right]+
	\log\binom{n}{ny^\star} 
	\right\},
\end{align}
where 
$n$ is a known constant,
$\varphi = 1$,
$\vartheta=\log(\mu/(1-\mu))$,
$b(\vartheta)=-\log(1-\mu)=\log(1+\mathrm{e}^\vartheta)$,
$c(y^\star,\varphi)=\log\binom{n}{ny^\star} $.
Therefore, 
$b_k^{\prime}(\vartheta)=\mathrm{e}^\vartheta/(1+\mathrm{e}^\vartheta)$
and
$b^{\prime\prime}(\vartheta)=\mathrm{e}^\vartheta/(1+\mathrm{e}^\vartheta)^2$
and
$c^{\prime}(y;\varphi)= c^{\prime\prime}(y;\varphi)= 0$. 
This implies that
$\mathrm{I\!E}(y)=\mu$,
$V(\mu) = \mu(1-\mu)$ and
$\mathrm{I\!E}(c_k^{\prime\prime}(y_{kt};\varphi_k))=0$.

\subsection{Negative Binomial}

The Negative Binomial is a discrete probability distribution often used to model overdispersed count data. Let $Y$ be a random variable following a Negative Binomial distribution. A commonly used parameterization for its probability function can be expressed as 
\begin{align}\label{eq_pdf_nb}
	f(y;\mu,\kappa)=\frac{\Gamma(\kappa+y)}{\Gamma(\kappa)y!}
	\left(
	\frac{\mu}{\mu+\kappa}
	\right)^y
		\left(
	\frac{\kappa}{\mu+\kappa}
	\right)^\kappa, \qquad y=0,1,2,\ldots,
\end{align}
where $\mu$ is the mean of $Y$ and $\kappa$ a precision parameter.

For representing~\eqref{eq_pdf_nb} in the canonical exponential family, we must assume that $\kappa$ is known. From this assumption, we can write \eqref{eq_pdf_nb} in the form of \eqref{eq_fam_exp_canonical} as
\begin{align*}
	f(y;\mu,\kappa)=\mathrm{exp}
	\left\{
	y\log 
	\left(
	\frac{\mu}{\mu+\kappa}
	\right)
	-\kappa\log(\mu+\kappa)+
	\kappa\log(\kappa)+\log\left(
	\frac{\Gamma(\kappa+y)}{\Gamma(\kappa)y!}
	\right)
	\right\}.
\end{align*}
Then, 
$\vartheta = \log 
\left(
\mu/(\mu+\kappa)
\right)$,
$\varphi = 1$,
$
b(\vartheta) = \kappa\log(\mu+\kappa)=
\kappa\left[\log(\kappa)-\log(1-\mathrm{e}^\vartheta)\right]
$,
and
$
c(y;\kappa)=c(y;\varphi)=	\kappa\log(\kappa)+\log\left(
\Gamma(\kappa+y)/\Gamma(\kappa)y!
\right)
$. 
The first and secord-order derivatives of $b(\cdot)$ are equal to 
$
b^\prime(\vartheta)= \kappa\mathrm{e}^\vartheta/(1-\mathrm{e}^\vartheta)
$
and
$
b^{\prime\prime}(\vartheta)= \kappa\mathrm{e}^\vartheta/(1-\mathrm{e}^\vartheta)^2
$,
respectively. 
From of $b^{\prime\prime}(\vartheta)$, we find that the variance function is $V(\mu)= \mu+\mu^2/\kappa$. Note that when $\kappa\rightarrow \infty$, $V(\mu)=\mu$ is the variance function of the Poisson distribution. Then, the Negative Binomial has a limiting Poisson distribution for $\kappa\rightarrow \infty$. On the other hand, since $\kappa$ is known, we have that 
$c^\prime(y;\kappa)=c^{\prime\prime}(y;\kappa)=0$.
Therefore, $\mathrm{I\!E}(c^{\prime\prime}(y;\kappa))=0$.

From these quantities, we can derive the expression for the score vector component in~\eqref{EQ_dell_dvarphik}, with the term $a_{kt}$ being expressed as 
\begin{align*}
	a_{kt}=-y_{kt}\log 
	\left(
	\frac{\mu_{kt}}{\mu_{kt}+\kappa}
	\right)
	+
	\kappa\log(\mu_{kt}+\kappa).
\end{align*}
Moreover, we also can obtain the component $\bm{K_{\varphi\varphi}}$ of the Fisher information matrix, computing the value of $d_{kt}$ in~\eqref{EQ_expect_dell_dvarphi2} as
\begin{align*}
	d_{kt}=2
	\left[
	\mu_{kt}\log 
	\left(
	\frac{\mu_{kt}}{\mu_{kt}+\kappa}
	\right)
	-
	\kappa\log(\mu_{kt}+\kappa)
	\right],
\end{align*}
since $\mathrm{I\!E}(c_k^{\prime\prime}(y_{kt};\varphi_k))=0$.

\subsection{Poisson}

The Poisson is a discrete probability distribution. Typically, it is used to model the occurrence of a certain number of events within a fixed time or space interval. It applies when these events happen independently at a constant, known average rate, regardless of the time since the last event. If $Y$ has a Poisson distribution, its probability function is given by
\begin{align}\label{eq_pdf_poisson}
	f(y;\mu)=\frac{\mu^y\mathrm{e}^{-\mu}}{y!}, 
\end{align}
where $y=0,1,2,\ldots$ is the number of occurrences and $\mu$ is the mean and variance of $Y$. That is, $\mathrm{I\!E}(Y)=\mathrm{var}(Y)=\mu$.

Writing~\eqref{eq_pdf_poisson} in canonical exponential family~\eqref{eq_fam_exp_canonical}, we obtain
\begin{align*}
	f(y;\mu) = \mathrm{exp}\{y\log (\mu)-\mu-\log(y!)\},
\end{align*}
where
$\vartheta = \log(\mu)$,
$\varphi=1$,
$b(\vartheta) = \mu=\mathrm{e}^\vartheta$,
and
$c(y;\varphi) = -\log(y!)$.
Hence, 
$b^{\prime}(\vartheta)=b^{\prime\prime}(\vartheta)=\mathrm{e}^\vartheta$
and
$c^{\prime}(y;\varphi)= c^{\prime\prime}(y;\varphi)= 0$. 
This implies that
$\mathrm{I\!E}(Y)=V(\mu) = \mu$ and
$\mathrm{I\!E}(c_k^{\prime\prime}(y_{kt};\varphi_k))=0$.

In this case, the term $a_{kt}$ in~\eqref{EQ_dell_dvarphik} is given by
\begin{align*}
	a_{kt} = \mu_{kt}-y_{kt}\log(\mu_{kt}).
\end{align*}
Similarly, we can obtain $d_{kt}$ of~\eqref{EQ_expect_dell_dvarphi2}
to
compute the component $\bm{K}_{\varphi\varphi}$ of the Fisher information matrix.
For this distribution, we have that
\begin{align*}
	d_{kt}=	2(\mu_{kt}\log(\mu_{kt})-\mu_{kt}).
\end{align*}

\subsection{Gamma}

The gamma is a continuous probability distribution often used in statistics to model the time it takes for events to occur. It is versatile and has two shape parameters, which allow it to take on various forms. 
The probability density function of $Y$ distributed as gamma in the generalized linear model's context is conveniently written in the form~\citep{mccullagh1989generalized}
\begin{equation}\label{EQ_pdf_gamma}
	f(y;\mu,\nu)=\frac{1}{y\,\Gamma(\nu)}
	\left(
	\frac{\nu\,y}{\mu}
	\right)^\nu
	\mathrm{exp}
	\left(
	-\frac{\nu\,y}{\mu}		
	\right), \qquad y>0,\,\,\nu>0,\,\,\mu>0.
\end{equation}
For brevity, we denote $Y\sim G(\mu,\nu)$.

The density~\eqref{EQ_pdf_gamma} can be written in the exponential family form given in~\eqref{eq_fam_exp_canonical} as
\begin{align*}
	f(y;\mu,\nu)=
	\mathrm{exp}
	\left\{
	\nu 
	\left[
	-\frac{y}{\mu}-\log \mu
	\right]
	+(\nu-1)\log y+\nu\log\nu-\log\Gamma(\nu)
	\right\}.
\end{align*}
Thus, in this case $\vartheta=-1/\mu$,
$\varphi = 1/\nu$,
 $b(\vartheta)=\log\mu=-\log(-\vartheta)$, and $c(y,\varphi)=(\nu-1)\log y+\nu\log\nu-\log\Gamma(\nu)=
 (\varphi^{-1}-1)\log y+\varphi^{-1}\log\varphi^{-1}-\log\Gamma(\varphi^{-1})
 $.

Computing the order second and first derivatives of $b(\vartheta)$, we have 
	$b^\prime(\vartheta)=-1/\vartheta=\mu$
	and
		$b^{\prime\prime}(\vartheta)=1/\vartheta^2=\mu^2$.
Hence, $\mathrm{I\!E}(Y)=\mu$, $V(\mu)=\mu^2$, and $\mathrm{var}(Y)=\varphi\mu^2$.

Since 
$c(y,\varphi)=
(\varphi^{-1}-1)\log y+\varphi^{-1}\log\varphi^{-1}-\log\Gamma(\varphi^{-1})
$, it follows that
$c^{\prime}(y,\varphi)=\varphi^{-2}[\log\varphi+\psi(\varphi^{-1})-\log y-1]$ 
and
$c^{\prime\prime}(y,\varphi)=
2\varphi^{-3}[\log y - \log \varphi - \psi(\varphi^{-1})+1]+\varphi^{-3}-\varphi^{-4}\psi^{(1)}(\varphi^{-1})$.
From these quantities, we obtain the expression of $a_{kt}$ of the score vector for $\varphi_k$ in~\eqref{EQ_dell_dvarphik} as

\begin{align*}
		a_{kt} =
	\left\{
	\varphi_k^{-2}
	\left[-\frac{1}{\mu_{kt}}y_{kt} - \log\mu
	\right] + 
	\varphi_{k}^{-2}[\log\varphi_{k}+\psi(\varphi_k^{-1})-\log y_{kt}-1]
	\right\}.
\end{align*}

Similarly, we need compute the value of $d_{kt}$ in~\eqref{EQ_expect_dell_dvarphi2} to obtain the component $\bm{K}_{\varphi\varphi}$ of the conditional Fisher information matrix. 
Let $Z=\log(Y)$. 
Since $Y\sim$Gamma($\mu,\varphi$), it can be shown that $\mathrm{I\!E}(Z)=\psi(\varphi^{-1})-\log (\varphi^{-1}) +\log\mu$.
Then,
\begin{align*}
	d_{kt}=	
(1-
2\varphi^{-3})[
\log\mu+1
]+\varphi^{-3}-\varphi^{-4}\psi^{(1)}(\varphi^{-1}).
\end{align*}

\subsection{Normal}

The normal distribution, a widely used probability distribution in statistics, is characterized by its symmetric, bell-shaped curve. What makes it particularly accessible is its definition, which is as simple as two parameters: the mean ($\mu$) and the standard deviation ($\sigma$).
If $Y$ follows a normal distribution, then its probability density function is defined as 
\begin{align*}
	f(y;\mu,\sigma)=\frac{1}{\sqrt{2\pi\sigma^2}}
	\mathrm{exp}\left\{-\frac{1}{2}\left(\frac{y-\mu}{\sigma}\right)^2\right\}, \qquad -\infty<y<\infty,
\end{align*}
that can be rewritten according to~\eqref{eq_fam_exp_canonical} as
\begin{align*}
	f(y;\mu,\sigma)=
	\mathrm{exp}\left\{
	\frac{1}{\sigma^2}\left(y\mu-\frac{\mu^2}{2}\right)-
	\frac{y^2}{2\sigma^2}-\frac{1}{2}\log(2\pi\sigma^2)
	\right\}, 
\end{align*}
being
$ \vartheta=\mu$,
$\varphi = \sigma^2$,
$b(\vartheta) =  \mu^2/2=\vartheta^2/2$,
$c(y;\varphi)=-y^2/(2\varphi)-1/2\log(2\pi\varphi)$.
Deriving the $b(\cdot)$ and $c(\cdot;\cdot)$ functions, we obtain
$b^\prime(\vartheta) =  \vartheta$,
$b^{\prime\prime}(\vartheta) =  1$,
$c^{\prime}(y;\varphi)=y^2/(2\varphi^2)-1/(2\varphi)$,
$c^{\prime\prime}(y;\varphi)=-y^2/\varphi^3+1/(2\varphi^2)$.
Thus, the variance function is $V(\mu)=1$.

Computing $a_{kt}$, it follows that
\begin{align*}
	a_{kt}= -\varphi_k^{-2}\left(
	y_{kt}\mu_{kt}-\frac{\mu_{kt}^2}{2}+\frac{y_{kt}^2}{2\varphi_k^2}-\frac{1}{2\varphi_k}
	\right).
\end{align*} 
For obtaining $d_{kt}$, it is necessary to compute the expected value of $c^{\prime\prime}(y;\varphi)$. Notice that, $\mathrm{I\!E}(Y^2)=\mu^2+\varphi$, since $Y\sim \mathcal{N}(\mu, \varphi)$. Hence, 
\begin{align*}
	\mathrm{I\!E}\left(c^{\prime\prime}(Y;\varphi)\right)=
	\frac{1}{2\varphi^2}-\frac{1}{\varphi^2}-\frac{\mu^2}{\varphi^3}.
\end{align*}
From these quantities, we compute the value of $d_{kt}$ as
\begin{align*}
	d_{kt}=2\varphi_k^{-3}\left(
\mu_{kt}^2-\frac{\mu_{kt}^2}{2}+	\frac{1}{2\varphi_k^2}-\frac{1}{\varphi_k^2}-\frac{\mu_{kt}^2}{\varphi_k^3}
	\right).
\end{align*}

\subsection{Inverse normal}

The inverse normal distribution is a two-parameter class of continuous probability distributions, useful for modeling data with support on positive real. Let $Y$ be a random variable following an inverse normal distribution. Then, its probability density function is given by
\begin{align}\label{eq_pdf_inverse_normal}
f(y;\mu,\phi)= \frac{\phi^{1/2}}{(2\pi y^3)^{1/2}}\mathrm{exp}
\left\{
-\frac{\phi(y-\mu)^2}{2\mu^2y}
\right\}, \qquad y>0,
\end{align}
where $\mu$ is the mean of $Y$ and $\phi>0$ is a shape parameter.

The density~\eqref{eq_pdf_inverse_normal}
can be expressed in canonical exponential family as
\begin{align*}
	f(y;\mu,\phi)=\mathrm{exp}
	\left\{
	\phi\left[
	y\left(
	-\frac{1}{2\mu^2}
	\right)-
	\left(
	-\frac{1}{\mu}
	\right)
	\right]
	-\frac{1}{2}
	\left[
	\log\left(
	\frac{2\pi y^3}{\phi}
	\right)+\frac{\phi}{y}
	\right]
	\right\},
\end{align*}
where
$\vartheta = -1/(2\mu^2)$,
$\varphi=\phi^{-1}$,
$b(\vartheta)=-1/\mu=-\sqrt{-2\vartheta}$,
and
$c(y;\phi)=	-1/2
\left[
\log\left(
2\pi y^3/\phi
\right)+\phi/y
\right]
$,
that is,
$
c(y;\varphi)=
-1/2
\left[
\log\left(
2\pi \varphi y^3
\right)+1/(\varphi y)
\right]
$. 
The first and second-order derivatives of the $b(\cdot)$ and $c(\cdot;\cdot)$ are
$b^\prime(\vartheta)=(-2\vartheta)^{-1/2}$,
$b^{\prime\prime}(\vartheta)=1/(-2\vartheta)^{3/2}$,
$c^\prime(y;\varphi)=1/(2y\varphi^2)-1/(2\varphi)$,
$c^{\prime\prime}(y;\varphi)=1/(2\varphi^2)-1/(y\varphi^3)$.
From the second-order derivative of $b(\cdot)$, we obtain the variance function $V(\mu)=\mu^3$. 

The term $a_{kt}$ is computed as
\begin{align*}
	a_{kt} = -\varphi_{k}^{-2}
	\left(
	-\frac{y_{kt}}{2\mu_{kt}^2}+\frac{1}{\mu_{kt}}+
	\frac{1}{2y_{kt}\varphi_k^2}-\frac{1}{2\varphi_k}
	\right).
\end{align*}
To obtain $d_{kt}$, we need to compute $\mathrm{I\!E}\left(c^{\prime\prime}(Y;\varphi)\right)$, which is given by
\begin{align*}
	\mathrm{I\!E}\left(c^{\prime\prime}(Y;\varphi)\right)=
	\frac{1}{2\varphi^2}-\frac{1}{\varphi^3\mu}-\frac{1}{\varphi^2},
\end{align*}
since it can be shown that $\mathrm{I\!E}(1/Y)=1/\mu+\varphi$. Therefore,
\begin{align*}
	d_{kt}=2\varphi_{k}^{-3}\left(
	-\frac{1}{2\mu_{kt}}+\frac{1}{\mu_{kt}}+\frac{1}{2\varphi_k^2}-\frac{1}{\varphi_k^3\mu_{kt}}-\frac{1}{\varphi_k^2}
	\right).
\end{align*}

\end{document}